\documentclass[journal,twoside,web]{ieeecolor}

\usepackage{generic}
\usepackage{cite}

\usepackage{amsmath,amssymb,amsfonts,amsthm,float}
\usepackage{algorithm}
\usepackage{algpseudocode}
\usepackage{graphicx}
\usepackage{textcomp}
 \usepackage{calc}
\usepackage{cases}
\usepackage{verbatim}
\usepackage{hyperref}
\usepackage{color, soul}
\usepackage{nccmath}
\usepackage{mathtools, nccmath}
\usepackage{dsfont}
\usepackage[normalem]{ulem}

\theoremstyle{plain}
\newtheorem{thm}{Theorem}

\newtheorem{prop}{Proposition}

\newtheorem{lemma}{Lemma}

\newtheorem{ex}{Example}

\theoremstyle{definition}
\newtheorem{defn}{Definition}

\newtheorem*{rem*}{Remark}

\theoremstyle{remark}

\newtheorem{fact}{Fact}

\def\arrvline{\hfil\kern\arraycolsep\vline\kern-\arraycolsep\hfilneg}

\newcommand{\R}{\mathbb{R}}
\newcommand{\N}{\mathbb{N}}
\renewcommand{\S}{\mathcal{S}}

\newcommand{\card}[1]{|#1|}

\renewcommand{\rm}[1]{\mathrm{#1}}

\newcommand{\ac}{\mathcal{A}} 
\newcommand{\p}{\mathcal{I}} 
\newcommand{\W}{\mathrm{W}} 
\newcommand{\U}{\mathrm{U}_i} 
\newcommand{\G}{\mathrm{G}} 

\newcommand{\br}{\text{Br}_i} 
\renewcommand{\ne}{a^{\mathrm{ne}}} 
\newcommand{\NE}{\mathrm{NE}} 
\newcommand{\pob}{\mathrm{Eff}} 
\newcommand{\poa}{\mathrm{PoA}} 
\newcommand{\emp}{a^{\varnothing}}  
\renewcommand{\k}{k}
\newcommand{\sol}{a^{\k}} 

\newcommand{\rr}{\mathcal{R}} 
\newcommand{\w}{w} 
\newcommand{\wf}{\tilde{w}} 
\newcommand{\ww}{\mathcal{W}}
\newcommand{\wfw}{\Gamma}
\newcommand{\wfwc}{\wfw_{\rm{id}}}
\newcommand{\setgmw}{\mathcal{G}_{\ww, \wfw}} 

\newcommand{\f}{u} 
\newcommand{\fw}{\mathcal{U}}  
\newcommand{\setgm}{\mathcal{G}_{\ww, \fw}} 
\newcommand{\fwc}{\fw_{\rm{CI}}}
\newcommand{\cc}{\mathrm{C}}  

\newcommand{\fmc}{\f^{\mathrm{mc}}} 
\newcommand{\wbc}{\w^{b, \cc}}  
\newcommand{\fa}{\f^{\infty}}

\newcommand{\abr}{a^{\mathrm{br}}}  
\newcommand{\aopt}{a^{\mathrm{opt}}}  

\newcommand{\ay}{a}
\newcommand{\xx}{x}
\newcommand{\bz}{b}
\newcommand{\paval}{\theta(\ay, \xx, \bz, \ell)}  
\newcommand{\paop}{\Theta(\ay, \xx, \bz, \ell)}  
\newcommand{\rraxb}{\rr_{\ay, \xx, \bz, \ell}^k} 

\newcommand{\wsc}{\w_{\mathrm{sc}}} 
\newcommand{\ipoa}{\mathcal{X}} 
\newcommand{\fxx}{\f^{\ipoa}} 

\def\BibTeX{{\rm B\kern-.05em{\sc i\kern-.025em b}\kern-.08em
    T\kern-.1667em\lower.7ex\hbox{E}\kern-.125emX}}
\begin{document}

\title{Best Response Sequences and Tradeoffs in Submodular Resource Allocation Games}
\author{Rohit Konda, \IEEEmembership{Student Member, IEEE}, Rahul Chandan \IEEEmembership{Student Member, IEEE}, David Grimsman \IEEEmembership{Member, IEEE}, Jason R. Marden \IEEEmembership{Fellow, IEEE}
\thanks{R. Konda (\texttt{rkonda@ucsb.edu}), R. Chandan, and J. R. Marden are with the Department of Electrical and Computer Engineering at the University of California Santa Barbara and D. Grimsman is with the Department of Computer Science at Brigham Young University. Preliminary results of this work were presented at the \emph{American Control Conference}, Atlanta GA, June 2022 \cite{konda2021balancing}.  This work is supported by \texttt{ONR grant \#N00014-20-1-2359}, AFOSR grants \texttt{\#FA9550-20-1-0054} and \texttt{\#FA9550-21-1-0203}, and the Army Research Lab through the \texttt{ARL DCIST CRA \#W911NF-17-2-0181}.}}

\maketitle

\begin{abstract}
Deriving competitive, distributed solutions to multi-agent problems is crucial for many developing application domains; Game theory has emerged as a useful framework to design such algorithms. However, much of the attention within this framework is on the study of equilibrium behavior, whereas transient behavior is often ignored.
Therefore, in this paper we study the transient efficiency guarantees of \emph{best response processes} in the context of \emph{submodular resource allocation games}, which find application in various engineering contexts. Specifically the main focus of this paper is on characterizing the optimal short-term system-level behavior under the best-response process. Interestingly, the resulting transient performance guarantees are relatively close to the optimal asymptotic performance guarantees.  Furthermore, we characterize the trade-offs that result when optimizing for both asymptotic and transient efficiency through various utility designs.
\end{abstract}

\begin{IEEEkeywords}
Game Theory, Optimization, Utility Design, Distributed Algorithms, Resource Allocation
\end{IEEEkeywords}

\section{Introduction}
\label{sec:int}

Multi-agent architectures have recently gained considerable attention due to their widespread applications including, but not limited to, swarm formation control \cite{chung2018survey}, e-commerce \cite{guttman1998agent}, traffic control \cite{wang2013distributed}, and smart grid pricing \cite{wang2014stackelberg}. The underlying goal of these distributed systems is to coordinate agents to desirable system states, as measured by some global objective, through local decision making processes. As such, \emph{game theoretic methods} have emerged as an important design methodology in these settings. In this approach, each agent is treated as a player in a non-cooperative game with their decision-making governed by a local-objective or \emph{utility} function.

\begin{figure}[h]
    \centering
    \includegraphics[width=250pt]{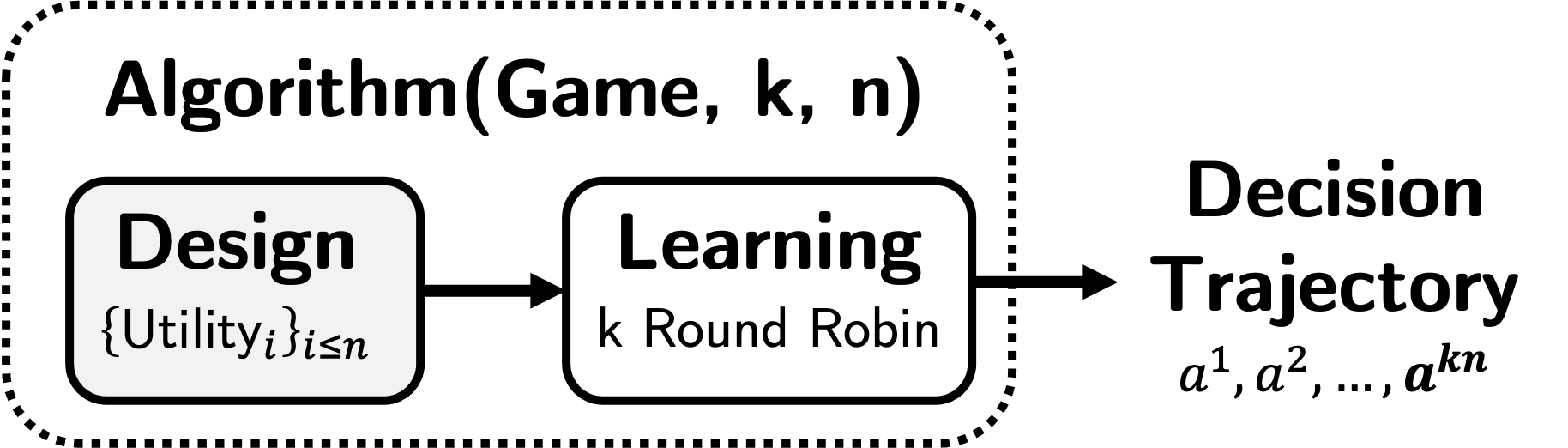}
    \caption{If a given multi-agent scenario with $n$ agents is modeled as a game, the construction of distributed algorithms can be decoupled into two domains: the design of local objectives (utilities) and the design of the learning dynamics. In this paper, we fix the dynamics to the classical $\k$ round-robin best response and study the effects of the utility design on the efficiency bounds for the resulting decision trajectory. Moreover, we characterize the guarantees as the number of rounds $\k$ increases.}
    \label{fig:intro}
\end{figure}

Endowed with a utility function, each agent can update its decisions in a self-interested manner to dynamically respond to environmental changes, including those induced by other agents. In this work, we consider when agents use the classical update rule - the \emph{round-robin best response} algorithm. Under this process, the decision updates are done in a sequential manner: at each iteration, a single, chosen agent optimizes its decision against its utility function with all of the other agents keeping their decisions fixed. As such, the utility structure of the agents can significantly influence the underlying dynamical behavior of the agents: this is highlighted in Figure \ref{fig:intro}. Understanding this relationship is important in classifying the emergent behavior, especially when a system operator would like to design or tune the utility functions in an optimal fashion.

The resulting limit points from running the best response process are exactly the set of \emph{Nash equilibrium}; in fact, the limit points of many different types of update rules correspond to Nash equilibrium \cite{young1993evolution}. Thus, within a game-theoretic framework, Nash equilibrium can operate as the fundamental model for characterizing end-behavior. This is the approach taken by much of the literature, where the limiting behavior of the round-robin best response algorithm against the global objective is characterized (and possibly optimized) through established performance metrics on the Nash equilibrium (for e.g., see \cite{roughgarden2009intrinsic,vetta2002nash, koutsoupias1999worst, paccagnan2019utility}).

While the many of the existing theoretical results on the Nash equilibrium behavior are positive, these performance guarantees only emerge asymptotically. In fact, arriving at Nash equilibrium may even take an exponential time \cite{fabrikant2004complexity}, rendering the resulting characterizations irrelevant in many realistic multi-agent scenarios. For example, there may be an extremely large number of agents in the multi-agent scenario or the relevant situational parameters may be time-varying and volatile or there may be computational and run-time restrictions on the agents. In these instances, expecting that the agents will converge to Nash equilibrium may not be a reasonable assumption.

Therefore, in this paper, we shift focus to the transient behavior of the round-robin best response algorithm. Specifically, we benchmark the iterative process to be the round-robin best response algorithm and study the performance guarantees that result from various designs of utility functions in the context of the well-studied class of \emph{submodular resource allocation games}. This class of games can model many relevant engineering applications \cite{marden2013distributed} and have convergence guarantees of the best response algorithm. The main results of this paper are as follows.

\begin{table}[h]
    \centering
    \begin{tabular}{c|c|c|c|}
        \multicolumn{1}{c}{} & \multicolumn{1}{c}{$\k = 1$} & \multicolumn{1}{c}{$\k > 1$} & \multicolumn{1}{c}{$\k = \infty$} \\
        \cline{2-4}
        $\wfw^*$ & $1 - \cc/2$ & $1 - \cc/2$ & $1 - \cc/e$ \\
        \cline{2-4}
        $\wfwc$ & $(1+\cc)^{-1}$ & $(1+\cc)^{-1}$ & $(1+\cc)^{-1}$ \\
        \cline{2-4}
        \multicolumn{1}{c}{} & \multicolumn{1}{c}{From \cite{konda2024optimal}} & \multicolumn{1}{c}{Theorem \ref{thm:kroundC}} & \multicolumn{1}{c}{Proposition \ref{prop:oneroundC}}
    \end{tabular}
    \caption{Comparison of Efficiency Results of the $\k$-round algorithm under the naive utility design $\wfwc$ and the optimal design $\wfw^*$ under different curvatures $\cc \in [0, 1]$.}
    \label{tab:eff}
\end{table}

\begin{itemize}
    \item Our first main result in Theorem \ref{thm:kroundC} provides the optimal transient performance guarantees for any finite $\k \geq 1$ rounds, which is summarized in Table \ref{tab:eff}. Interestingly, we note that these performance guarantees do not increase past the first round. This non-intuitive result suggests minimal benefit in performance when allowing agents to alter their decisions in response to other agents.
    
    \item In general, utility designs that optimize transient guarantees may not be optimal for the asymptotic behavior (and vice-versa). Our second result confirms this trade-off in submodular resource allocation games; in fact, in Theorem \ref{thm:submodtrade}, it is shown that optimizing for asymptotic performance guarantees can result in arbitrarily bad transient performance guarantees. Furthermore, we delineate the exact Pareto-optimal frontier for the asymptotic and transient performance guarantees in Theorem \ref{thm:poapobtradeoff} for a subclass of set covering games. 
\end{itemize}
In previous work \cite{konda2024optimal}, the efficiency guarantees for $\k=1$ rounds were characterized for resource allocation games. In this work, we extend these results to consider the efficiency guarantees for finite rounds $k\geq 1$ and asymptotic efficiency $k = \infty$ as well as consider the trade-offs between optimizing for one-round and asymptotic efficiency. The results of this work suggest an acute diminishing return on the achievable performance guarantees as agents make successive locally optimal decisions. While the current literature is sparse on characterizing short-term behavior of game-processes, the non-intuitive results in this work indicate the importance of analyzing the interaction between local objective designs, asymptotic performance, and transient performance as a whole.

Moreover, this work belongs to a larger research trend that aims to study game theoretic models beyond their respective equilibrium. In contrast to the traditional game-theoretic approach, the \emph{game dynamics} are embraced as a valuable feature of the game, where a rigorous study of actualized play can provide important insights about the game model (see, for e.g., \cite{goemans2005sink, candogan2011flows, papadimitriou2019game}). Furthermore, characterizing performance guarantees along these game dynamics is valuable to understanding the transient behavior of the agents. In contrast to the study of equilibrium quality, the literature on transient guarantees is much less developed. However, we highlight an important subset of works that characterize transient performance guarantees in different game-theoretic contexts: such problem domains include affine congestion games \cite{christodoulou2006convergence, bilo2009performances, fanelli2008speed, bilo2018unifying}, market sharing games \cite{christodoulou2006convergence, mirrokni2004convergence}, basic utility games \cite{mirrokni2004convergence}, series-parallel networks, and load-balancing games \cite{suri2007selfish, caragiannis2006tight}.

\section{Model}
\label{sec:model}

We consider distributed settings that are modeled as \emph{submodular resource allocation games} \cite{marden2013distributed}. Let $\p = \{1, \dots, n\}$ be the collection of agents that can possibly utilize a portion of a given finite set of resources $\rr = \{r_1, \dots, r_{d}\}$. Each resource $r \in \rr$ is associated with a \emph{welfare rule} $\w_r :\N \to \R_{>0}$ that defines the welfare accrued at each resource based on number of agents that utilize it. In this work, we assume that the welfare rule $\w_r$ is non-decreasing and concave (or that $\w_r(j+1) - \w_r(j)$ is non-negative and non-increasing in $j$) and use the denotation $\w_r(0) = 0$. In this way, additional utilization of resources by the agents has positive but decreasing marginal return for the global objective. The choice of resource utilization for each agent is given by its action set  $\ac_i \subseteq 2^\rr$. As such, the quality of a joint action $a = (a_1, \dots, a_n) \in \ac = \ac_1 \times \cdots \times \ac_n$ is classified through a system-level objective function $\W: \ac \to \R_{> 0}$ of the form
\begin{equation}
    \W(a) = \sum_{r \in \rr}{\w_r(|a|_r)},
\end{equation}
where $|a|_r = \card{\{i \in \p : r \in a_i\}}$ is the number of agents that utilize $r$ in the joint action $a$. These types of objectives are commonplace in many engineering domains, with applications in information gathering, image segmentation, statistical summarization (see \cite{krause2014submodular}). We include a sample of these applications below.

\begin{ex}[Weapon-Target Assignment]
\label{ex:WTA}
Consider a weapon-target assignment problem \cite{manne1958target}, in which a set of agents $\p$ defend against a set of $\rr$ targets. Each agent $i \in \p$ has to decide which targets $r \in \rr$ to defend against with a $p_d$ chance of defending against each target, where its decision $a_i \subset \rr$ is a subset of targets. Each target is also characterized by its relative value $v_r \geq 0$. As a whole, the set of agents would like to maximize the expected value of the targets defended. As such, the nonlinear objective of the agent's decisions is denoted by
\begin{equation}
    \W(a) = \sum_{r \in \rr} v_r \cdot \left(1 - (1 - p_d)^{|a|_r}\right).
\end{equation}

\vspace{.5em}
\noindent \textnormal{\textbf{Example 1A} (Set Covering)\textbf{.}} If the probability of defending $p_d$ is $1$, then we recover the problem of set covering \cite{gairing2009covering}. In this domain, there is no benefit for more than agent to defend against a target, and as a whole, the agents would like to maximize the value of the targets defended. Thus, the welfare function simplifies to 
\begin{equation}
    \W(a) = \sum_{r \in \bigcup_i a_i} v_r.
\end{equation}
\end{ex}

\begin{ex}[Wireless Transmission over a Network]
\label{ex:wirtrans}
Consider a group of communicating agents, as in \cite{marden2012price}, that send transmissions through a shared network with nodes $V$ and edges $E$. Each agent would like to send a wireless transmission over the network from a given start node $s_i$ to an end node $t_i$, and it must choose one out of the possible paths from $s_i$ to $t_i$ to transmit across. Thus, the resource set is $E$ and the action set $\ac_i \subset 2^E$ is the possible set of paths from $s_i$ to $t_i$. An edge $e$ may experience congestion if multiple transmissions utilize it; we assume that with each additional transmission, the rate of transmission experiences a harmonic-like decay. As a whole, the agents would like to send their transmissions with the highest rate. As such, the system welfare is
\begin{equation}
    \W(a) = \sum_{e \in \bigcup_i a_i} \sum_{j=1}^{|a|_e} \frac{1}{j}.
\end{equation}
\end{ex}

The global directive of the agents is to coordinate to a joint action that maximizes the system welfare, i.e. $\aopt \in \arg \max_{a \in \ac} \W(a)$. In order to coordinate the agents in a distributed fashion, we assume a game-theoretic setup, where each agent optimizes its decision with respect to a given local objective, or utility function  $\U : \ac \to \R$ in a self-interested process. To establish the learning procedure of the agents, we focus attention to a class of best response processes known as \emph{$\k$-round walks} (or $\k$ round-robin best response), explicitly stated in Algorithm \ref{alg:cap2}. At each step of the algorithm, an agent is selected in a round-robin fashion to perform a best response; this goes on until $\k$ rounds have been completed. For a given joint action $\alpha \in \ac$, we say the action $\abr_i$ is a \emph{best response} for agent $i$ if 
\begin{equation}
\abr_i \in \ \br(\alpha_{-i}) = \arg \max_{a_i \in \ac_i} \U(a_i, \alpha_{-i}),
\end{equation}
where $a_{-i} = (a_1, \dots, a_{i-1}, a_{i+1}, \dots a_n)$ denotes the joint action $a$ without the action of agent $i$. We also assume that the round-robin walk begins with none of resources being utilized by any of the agents, denoted by the null joint action $\emp := \varnothing$.

\begin{algorithm}
\caption{$\k$-Round Walk}
\label{alg:cap2}
\begin{algorithmic}
\Require $a(0) \gets \emp$, $i \gets 1$, $\tau \gets 1$, $\k$, $n$
\While{$\tau \leq \k n$}
\State Modify action of agent $i$ to $a_i(\tau) \gets \br(a_{-i}(\tau-1))$;
\State Fix other agent actions to $a_{-i}(\tau) \gets a_{-i}(\tau-1)$;
\State Increment $\tau$ by $1$;
\State Set the next agent $i \gets \tau \mod n$ by round-robin;
\EndWhile
\State \Return $a(\k n)$
\end{algorithmic}
\end{algorithm}

Running the $\k$-round walk algorithm induces the action trajectory
$a(0)=\emp, a(1), \dots, a(\k n - 1 ), a(\k n)$ with $a(\k n)$ being the end resulting joint action. The central goal of this work is to understand how the utility functions of the agents affect the behavior of the $\k$-round walk. To this end, the performance of the joint action $\sol \doteq a(\k n)$ after $\k$ rounds can be characterized through the following metric \footnote{We note that the resulting action $a(\k n)$ may not be unique if the best response $\br(a_{-i}(\tau))$ is not unique for some $\tau$. In this case, we overload $\W(\sol)$ to mean the minimum welfare $\min_{\alpha \in \sol} \W(\alpha)$ for the $\k$-round walk.}
\begin{equation}
\label{eq:effG}
\pob(\G; \k) = \frac{\W(\sol)}{\max_{a \in \mathcal{A}}{\W(a)}} \in [0, 1],
\end{equation}
where a ratio closer to $1$ implies the efficiency of the joint action after $\k$ rounds is closer to optimal. We use tuple $\G \triangleq (\p, \ac, \rr, \W, \{\U\}_{i \in \p})$ to define the game instance under consideration. Furthermore, we characterize the performance of the limit points \footnote{Since resource allocation games are \emph{potential games} \cite{shapley1953value}, the limit points of $\k$-round walks are necessarily a subset of the set of Nash equilbrium.} of the $\k$-round walk as 
\begin{equation}
    \pob(\G; \infty) = \frac{\W(\lim_{\k \to \infty} \sol)}{\max_{a \in \mathcal{A}}{\W(a)}}.
\end{equation}

Accordingly, the form of the utility functions has a significant impact on the given metric $\pob(\G; \k)$. One natural utility function to consider is the \emph{common-interest} (CI) utility, where all the agents share the same utility function $\U(a) \equiv \W(a)$ for all $a \in \ac$ and $i \in \p$. In line with literature on submodular optimization \cite{nemhauser1978analysis}, we demonstrate that common interest utilities exhibit constant factor efficiency guarantees $\pob(\G; \k)$. However, as unveiled by subsequent results in this work, by finely adjusting the utility functions, we can achieve superior efficiency guarantees. Following this perspective, lets consider utility functions of the form
\begin{equation}
\label{eq:utildef}
    \U(a_i, a_{-i}) = \sum_{r \in \rr} \wf_r(|a|_r),
\end{equation}
where the \emph{distribution rule} $\wf_r : \N \to \R_{>0}$ defines the resource-specific agent utility determined by $|a|_r$. In general, $\wf_r \neq \w_r$ may not be equivalent \footnote{It is possible to introduce locality to Eq. \eqref{eq:utildef}, where value is summed across agent selections $r \in a_i$ instead of $r \in \rr$. We discuss this in the Appendix.}. We also assume that $\wf_r$ is non-decreasing and concave with $\wf_r(1) = \w_r(1)$ for all $r \in \rr$. Now, the choice of the distribution rules $\wf_r$ influences the resulting joint action trajectory, and in turn the resulting end joint action $\sol$.

In most scenarios of interest, a system designer is required to specify the distribution rules $\wf_r$ without specific knowledge of the resource allocation game parameters, such as the number of agent $\p$ or the action set $\ac$. To that end, let $\ww$ be the set of possible welfare rules that could be associated with any resource, i.e., $\w_r \in \ww$ for all  $r \in \rr$.  Here, the system designer is tasked with associating a distribution rule to each type of resource, i.e.,  the distribution rule for any resource $r\in \rr$ with the welfare rule $\w_r$ is of the form $\wf_r = \wfw(\w_r)$  where we refer to the map $\wfw: \ww \to \R^{\N}_{\geq 0}$ as the \emph{utility design}. We refer to the common interest utility design as $\wfwc$, where $\wfwc(\w_r) = \w_r$ for all $\w_r \in \ww$. Lastly, we define the set of resource allocation games that are induced by $\ww$ and $\wfw$ as $\setgmw$, where a game $\G \in \setgmw$ if $\w_r \in \ww$ and $\wf_r = \wfw(\w_r)$ for all resources $r \in \rr$.

\begin{ex}[Set Covering continued] When we reconsider the set covering problem in Example \ref{ex:WTA}A, the impact of different utility designs on the resulting decision trajectories can be highlighted through a simple game example, shown in Figure \ref{fig:ex3}. Let $\G$ be a game with two agents and three resources $\rr = \{r_1, r_2, r_3\}$ with values $v_1 = 1$, $v_2 = 1 + \varepsilon$, and $v_3 = \varepsilon$ for some $\varepsilon \ll 1$. The action set for the first agent is $\ac_1 = \{\varnothing, \{r_1\}, \{r_2\}\}$ and $\ac_2 = \{\varnothing, \{r_2\}, \{r_3\}\}$ for the second.

\begin{figure}[ht]
    \centering
    \includegraphics[width=120pt]{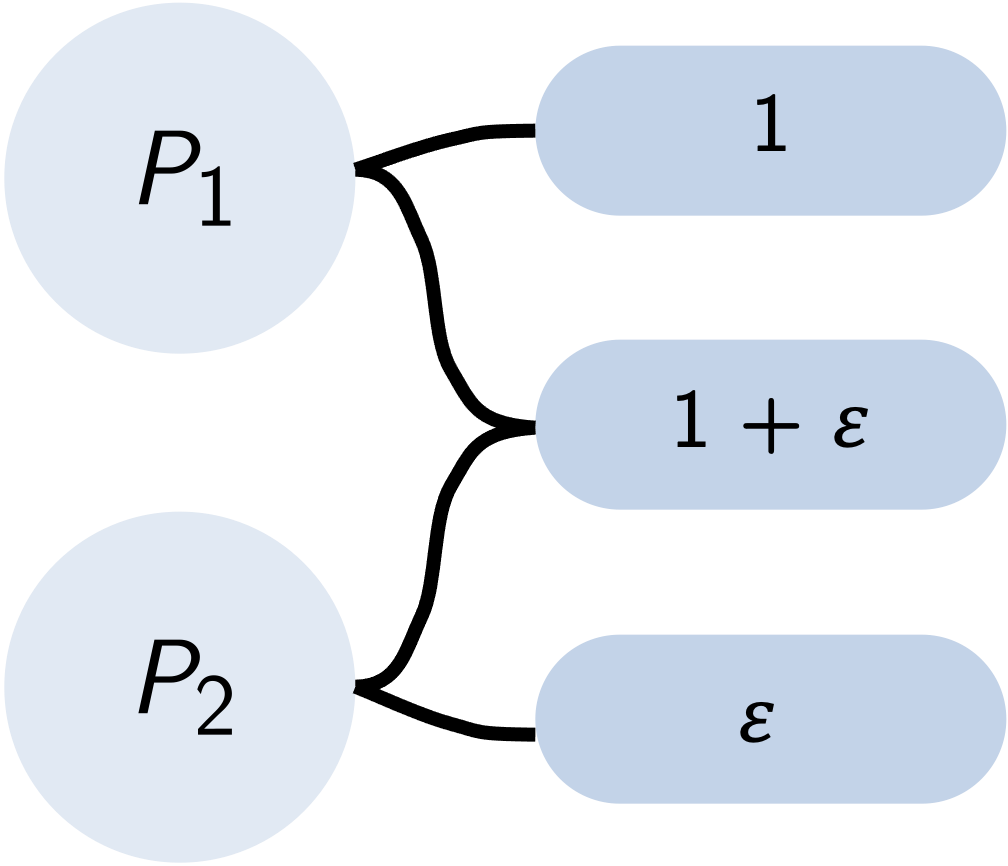}
    \caption{Game example depicting resource values and agent selections.}
    \label{fig:ex3}
\end{figure}

Under the common interest utility design $\wfwc$, the best response process will produce the resource selection trajectory $(\varnothing, \varnothing), (r_2, \varnothing), (r_2, r_3), (r_2, r_3), \dots$, resulting in a total welfare of $1 + 2 \varepsilon$ as $k \to \infty$. If we use instead the distribution rule $\wf \simeq [0, 1, 3/2, 3/2, \dots]$, the resource selection trajectory is then $(\varnothing, \varnothing), (r_2, \varnothing), (r_2, r_2), (r_1, r_2), (r_1, r_2), \dots$, resulting in a total welfare of $2 + \varepsilon$ as $k \to \infty$. As displayed, implementing the correct utility design (which may substantially differ from $\wfwc$) can result in significant improvements in the welfare values along the decision trajectory.
\end{ex}

The central focus of this paper is to understand how the choice of distribution rules, derived from $\wfw(\cdot)$, and the number of rounds $\k$ impacts the efficacy of the emergent collective behavior in the $\k$-round walk. Therefore, we additionally extend the efficiency measure to the set of games $\setgmw$ and quantify the optimal efficiency guarantees as
\begin{align}
\label{eq:effsetG}
    \pob(\ww, \wfw; \k) &= \inf_{\G \in \setgmw} \pob(\G; \k), \\
    \pob^*(\ww; \k) &= \sup_{\wfw: \ww \to \R^{\N}_{> 0}} \pob(\ww, \wfw; \k).
\end{align}

We similarly extend the definitions of $\pob(\G; \infty)$. We importantly highlight that $\lim_{\k \to \infty} \pob(\ww, \wfw; \k) \neq \pob(\ww, \wfw; \infty)$ in general, as the limit $\lim_{k \to \infty}$ and the infimum $\inf_{\G \in \setgmw}$ cannot be interchanged. We devote the next section to characterizing the introduced efficiency metrics.

\section{Main Results}
\label{sec:results}

Our work is concerned with classifying the guarantees associated with the $\k$-round walk. Before we introduce our results, we introduce a parametrization for the class of submodular welfare rules, denoted as \emph{curvature} \cite{conforti1984submodular, chandan2021tractable}. Curvature is a classical parametrization used extensively in the submodular optimization literature to describe the degree of submodularity of a given welfare function, where a curvature $\cc = 0$ indicates that a welfare rule is close to linear, and a curvature $\cc = 1$ indicates that a welfare rule is maximally submodular. We define curvature formally below.

\begin{defn}[Curvature]
A submodular welfare rule $\w$ has a curvature of $\cc \in [0, 1]$ if $\cc = 1 - \lim_{n \to \infty} (\w(n+1) - \w(n))/\w(1)$.
\end{defn}

With this, we can characterize the optimal efficiency guarantees for the $\k$-round walk (when $\k=1$ and when $\k=\infty$) as well as the efficiency guarantees for the common-interest utility below.

\begin{prop}
\label{prop:oneroundC}
Let the set $\ww$ comprise of all submodular welfare rules $\w$ that have curvature of at most $\cc \in [0, 1]$. The optimal efficiency guarantees are
\begin{align}
    \label{eq:OPTone}
    \pob^*(\ww; 1) &= 1 - \cc/2, \\
    \label{eq:OPTinfty}
    \pob^*(\ww; \infty) &= 1 - \cc/e.
\end{align}
and the guarantees associated with the common interest are
\begin{equation}
    \label{eq:CIoneinfty}
    \pob(\ww, \wfwc; 1) = \pob(\ww, \wfwc; \infty) = \left( 1 + \cc \right)^{-1}.
\end{equation}
\end{prop}

\begin{rem*}
We remark that Eq. \eqref{eq:OPTone} and the $\k=1$ characterization in Eq. \eqref{eq:CIoneinfty} come directly from prior work in  \cite[Theorem 3]{konda2024optimal}. The results in Eq. \eqref{eq:OPTinfty} and the $\k=\infty$ characterization in Eq. \eqref{eq:CIoneinfty} come from non-trivially extending the price of anarchy results found in \cite[Theorem 1]{chandan2021tractable}. Details on this extension is found in the Appendix.
\end{rem*}

We remark that the optimal efficiency guarantees after one round are relatively close to the optimal asymptotic guarantees\footnote{Furthermore, when utilizing the common interest utility design, the efficiency guarantees for $\k=1$ and $\k=\infty$ are identical.}. However, while small, there is still a gap between the optimal efficiency guarantees for $\k=1$ and $\k=\infty$. We expect this efficiency gap to decay as we run the $\k$-round walk for more rounds. However, we show surprisingly that further rounds do not increase the relative efficiency guarantees. Specifically, running the $\k$-round walk \emph{can not} improve the resulting efficiency guarantee for any given $\k$ over the guarantee of the one-round walk. This is made formal in the upper bound characterization stated in the next theorem.

\begin{thm}
\label{thm:kroundC}
Let the set $\ww$ comprise of all submodular welfare rules $\w$ that have curvature of at most $\cc \in [0, 1]$.  Then the efficiency guarantees of the $\k$-round walk, for any $\k \geq 1$, is upper bounded by the expressions
\begin{align}
    \label{eq:klessC}
    \pob^*(\ww; \k) &\leq \pob^*(\ww; 1) = 1 - \cc/2, \\
    \label{eq:klessCCI}
    \pob(\ww, \wfwc; \k) &= \pob(\ww, \wfwc; 1) = (1 + \cc)^{-1},
\end{align}
respectively for the the optimal utility design and the common interest utility.
\end{thm}

\begin{rem*}
We remark that the results in Theorem \ref{thm:kroundC} are not endemic to the $\k$-round walk algorithm. Allowing for variable turn order or only allowing strict best responses does not affect the resulting upper bounds. This is discussed in more detail in the Appendix in Section \ref{subsec:thmkapp}.
\end{rem*}

Notably, for any curvature $\cc \in [0, 1]$, the upper bound in Eq. \eqref{eq:klessC} exactly matches the characterization in Eq. \eqref{eq:OPTone} and likewise for the upper bound in Eq. \eqref{eq:klessCCI} and the characterization in Eq. \eqref{eq:CIoneinfty}. Therefore, in regards to the efficiency guarantees, running the $\k$-round walk algorithm for more than one round does not lead to gains in performance. We verify this result in the illustrative example described in Section \ref{sec:illex}.

In Proposition \ref{prop:oneroundC}, we describe the optimal performance guarantees for $\k=1$ and $\k=\infty$. Accompanying these results are characterizations of the utility designs that achieve said guarantees. As these utility designs are not equivalent, this prompts the natural question: Does optimizing the efficiency guarantees for $\k=1$ have down-stream effects on the efficiency for $\k=\infty$, and vice versa? This question is precisely  addressed in the next theorem. Explicitly, we identify the reciprocal guarantees for both the utility design that optimizes the transient guarantees for $\k=1$ and the utility design that optimizes the asymptotic guarantees for $\k=\infty$.

\begin{thm}
\label{thm:submodtrade}
Let the set $\ww$ comprise of all submodular welfare rules $\w$ that have curvature of at most $\cc \in [0, 1]$. Consider $\wfw^*_{\infty}$ to be the utility design that achieves the optimal $\pob^*(\ww; \infty)$ and $\wfw^*_1$ to be the utility design that achieves the optimal $\pob^*(\ww; 1)$. Then we have that
\begin{align}
    \label{eq:tradfrontone}
    \pob(\ww, \wfw_1^*; \infty) &= \pob^*(\ww; 1) = 1 - \cc/2\\
    \label{eq:tradfrontpoa}
    \pob(\ww, \wfw^*_{\infty}; 1) &\leq 1+\frac{\left(\cc-3\right)\cc}{\left(2-\cc\right)e+\cc}
\end{align}
with $\pob(\ww, \wfw^*_{\infty}; 1) < \pob^*(\ww; 1)$ holding strictly for $\cc > 0$ and $\pob(\ww, \wfw^*_{\infty}; 1) = 0$ for curvature $\cc=1$.
\end{thm}

We first observe that there are no gains in the asymptotic guarantees of the utility design $\wfw_1^*$ over the respective optimal guarantees for $\k=1$. Again, we see that there are diminishing returns of running the $\k$-round algorithm for more than one round. Additionally, the transient guarantees of $\wfw^*_{\infty}$ are strictly less than $\pob^*(\ww; 1)$ as expected. However, if the curvature $\cc=1$ is maximal, we note that transient guarantees of $\wfw^*_{\infty}$ unexpectedly degrade to $0$. Interestingly, optimizing for asymptotic performance does not necessarily translate to good transient performance in our setting. Moreover, it may even result in highly undesirable behavior in the transient in certain settings.

To clarify this stark trade-off between the transient and asymptotic guarantees, we restrict attention to the class of \emph{set covering games} \cite{gairing2009covering} (see Example \ref{ex:WTA}.A) and characterize the exact Pareto optimal frontier. Set covering games are natural generalizations of covering problems  \cite{gairing2009covering}, and are characterized by the following welfare rule (with curvature $\cc=1$). 
\begin{equation}
\label{eq:wscdef}
    \wsc(j) = \left\{\begin{array}{lr}
        1, & \text{for } j \geq 1\\
        0, & \text{for } j = 0\\
        \end{array}\right\}.
\end{equation}

\noindent With this, we arrive at the following Pareto frontier characterization, depicted in Figure \ref{fig:tradepareto}. Note that the end points of the trade-off curve matches the ones dictated in Theorem \ref{thm:submodtrade} for curvature $\cc=1$ exactly.

\begin{thm}
\label{thm:poapobtradeoff}
Let $\ww = \{\wsc\}$, where $\wsc$, defined in Eq. \eqref{eq:wscdef}, is the set covering welfare rule and $\wfw(\wsc) = \wf$ is the corresponding distribution rule. If the efficiency of the limit point of the $\k$-round walk is $\pob(\wsc, \wf; \infty) = Q \in [\frac{1}{2}, 1 - \frac{1}{e}]$ \footnote{We will sometimes use abuse of notation $\pob(\w_r, \wf; \k)$ to mean $\pob(\ww = \{\w_r\}, \wfw; \k)$ with $\wfw(\w_r) = \wf$ for a specific $\w_r$ and $\wf_r$.}, the maximum efficiency $\max_{\wf} \pob(\wsc, \wf; 1)$ achievable after $\k=1$ rounds is
\begin{equation}
\label{eq:tradeoffsetcov}
\left[ \sum^\infty_{j=0} \max \left\{ j! (1 - \frac{1 - Q}{Q} \sum^j_{\tau=1}{\frac{1}{\tau!}}) , 0\right\}+1 \right]^{-1}.
\end{equation}
\end{thm}

\begin{figure}[ht]
    \centering
    \includegraphics[width=250pt]{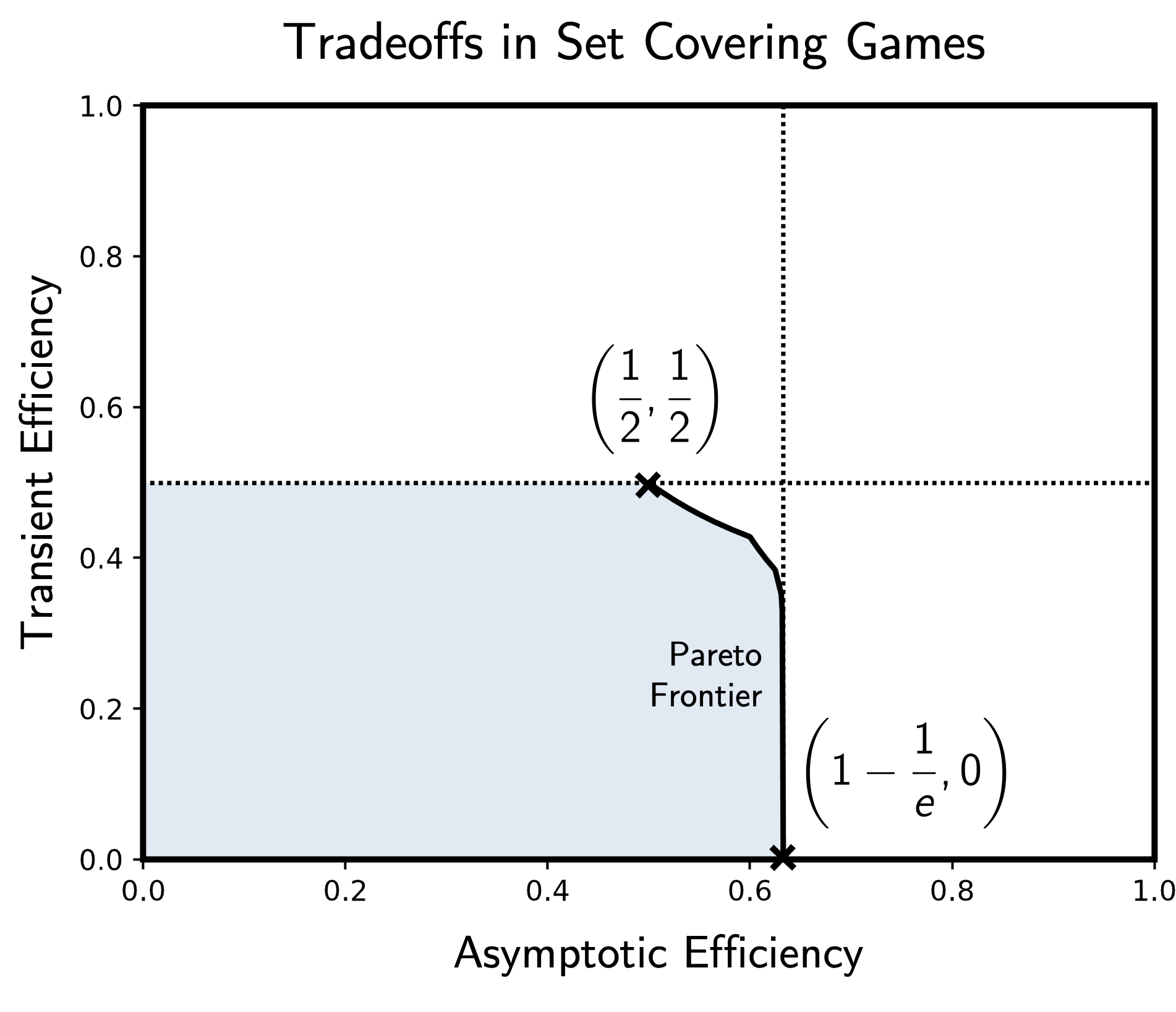}
    \caption{We depict the Pareto-optimal frontier of the one-round efficiency $\pob(\wsc, \wf; 1)$ versus the asymptotic efficiency guarantees $\pob(\wsc, \wf; \infty)$ that are possible with regards to the class of set-covering games. We note that the severe drop off in transient efficiency that results from optimizing the asymptotic efficiency.}
    \label{fig:tradepareto}
\end{figure}

Notably in Figure \ref{fig:tradepareto}, we see a stark drop-off in transient efficiency when the asymptotic efficiency is close to the optimal guarantee of $1 - \frac{1}{e}$. This extreme trade-off should prompt a more careful interpretation of asymptotic results, especially in the setting of resource allocation games.

\section{Illustrative Example}
\label{sec:illex}

To illustrate the results of the paper, we examine the average performance over $5$ rounds of the $\k$-round algorithm of three utility designs: the common interest utility design, the utility design that optimizes the efficiency for $\k=1$, and the utility design that optimizes the efficiency for $\k=\infty$. The average performance is measured across $100$ random instances of \emph{weapon-target assignment problems} (see Example \ref{ex:WTA}) with $20$ agents with a defense rate of $p_d = .5$. In each simulated instance, we set the number of targets that the agents can possibly defend to $30$. The values $v_r$ for each target $r$ are uniformly selected from the unit interval $[0, 1]$ and subsequently normalized by dividing by $\sum_{r \in \rr} v_r$. Each agent has $2$ actions available, in addition to the empty allocation $\emp$. Each action $a$ is a consecutive selection of $2$ resources chosen uniformly randomly from the resource set $\rr$. 

The resulting system welfare across $5$ rounds for each utility design is highlighted in Figure \ref{fig:numerical}, where the distributions of the system welfare across the randomized instances are depicted with a box and whisker plot. Note that the optimal allocation may also not achieve a $100\%$ detection rate. In Figure \ref{fig:numerical}, we see that worst instance of the optimal one-round performs better than the greedy and asymptotically optimal utility designs when $\k=1$. This is supported in the worst-case analysis presented in this paper. Additionally, we note that the resulting efficiency plateaus quickly, with almost no differences in efficiency after two rounds of best response - confirming that successive rounds give diminishing returns in system performance. Interestingly, on average, the differences in performance across utility designs is much more subtle.

\begin{figure}[ht]
    \centering
    \includegraphics[width=250pt]{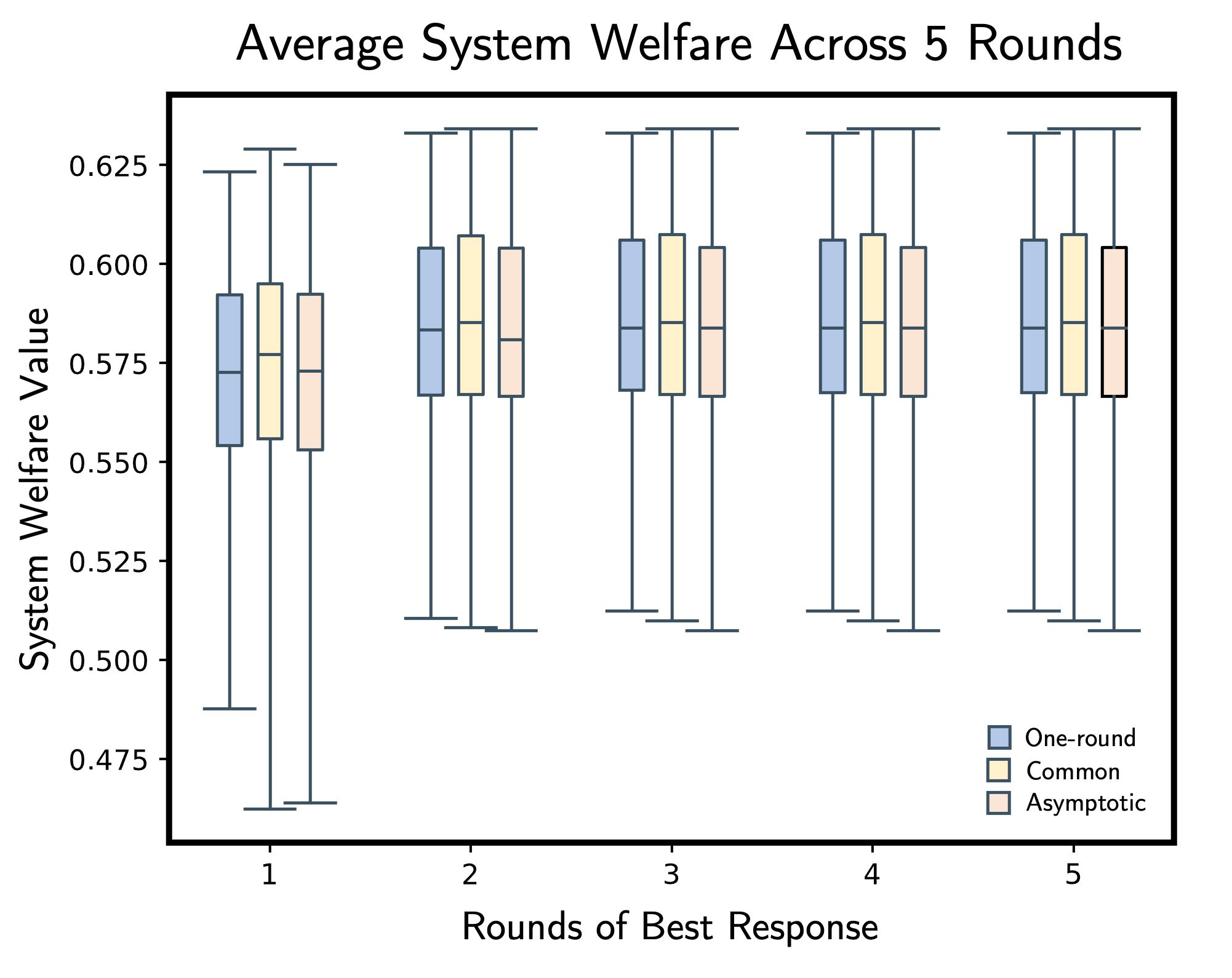}
    \caption{We plot the average rate of defense in a randomly generated set of weapon-target assignment problems with respect to three utility designs: the one-round optimal, the common interest, and the asymptotically optimal utility design. We see that in the short term, the one-round optimal design performs better in the worst case than the greedy and the asymptotically optimal utility designs.}
    \label{fig:numerical}
\end{figure}

\section{Conclusion}
\label{sec:conc}

This work focuses on characterizing the performance along a $\k$-round walk algorithm in submodular resource allocation problems. Using the language of game theory, this work shifts the emphasis from studying the system-level performance of the limit points of the $\k$-round walk to the transient states along each round of the algorithm. We can characterize the gap between the performance after one round and in the limit of the $\k$-round walk algorithm. We show that non-intuitively, this gap persists for any fixed finite number of rounds $\k$. Furthermore, we characterize the trade-offs when trying to optimize the utility functions to maximize the performance of the one-round and the limit points. Surprisingly, when optimizing for the asymptotic performance, the performance of the one-round walk may be arbitrarily bad. We additionally characterize the exact Pareto-frontier for these trade-offs in a sub-class of submodular resource allocation games. Future work may comprise of extending the results to other game models or consider average-case efficiency analysis.

\bibliographystyle{ieeetr}
\bibliography{references.bib}

\appendix
\label{sec:appendix}

We provide technical details for the proofs of the main results, as well as supporting discussion and lemmas, in the appendix. An overview of the provided proofs are as follows.

\begin{itemize}
    \item [(A)] We first introduce alternative formulations of the utility functions in Eq. \eqref{eq:utildef} to incorporate locality.
    \item [(B)] The connection between $\infty$-round walks and Nash equilibrium are established in Lemma \ref{lem:poblesspoa}, which allows for established results in price of anarchy to be used for asymptotic characterization.
    \item[(C)] We provide a proof of Proposition \ref{prop:oneroundC} that is based on previous work and results in price of anarchy.
    \item[(D)] A proof of Theorem \ref{thm:kroundC} is given, where game constructions provide the efficiency bounds for $\k \geq 1$ rounds. 
    \item[(E)] A proof of Theorem \ref{thm:submodtrade} is provided, where trade-offs are established for asymptotic and transient efficiency guarantees.
    \item[(F)] We provide a proof of Theorem \ref{thm:poapobtradeoff}, characterizing the exact Pareto-optimal frontier of the trade-offs for set covering welfare rules.
\end{itemize}

\noindent \emph{Notation.} Given a set $\S$, $\card{\S}$ represents its cardinality. We use the denotation $\w(0) = \f(0) = 0$. We also assume without loss of generality in the appendix that $\w_r(1) = 1$ for any resource $r$ \footnote{Consider any game $\G \in \setgm$. For any resource $r \in \rr$ with a welfare rule $\w_r$, we can define another game $\G' \in \setgm$ with instead $\card{\w_r(1)}$ copies of resource $r$ with a welfare rule $\w_r'(j) = \w_r(j)/\w_r(1)$. If $\card{\w_r(1)}$ is not integer, we can scale the number of resources uniformly and round to get arbitrarily close. Thus we can assume without loss of generality that $\w_r(1) = 1$.}. We define the bent welfare rule for some $b \geq 1$ and curvature $\cc \in [0, 1]$ as
\begin{equation}
\label{eq:wbc}
    \wbc(j) = (1-\cc) j + \cc \cdot \min \{j, b\}.
\end{equation}

\subsection{Equivalence of Marginal Contribution Utilities}
\label{subsec:MC}
We outline a key equivalence in order to reformulate the utility functions in Eq. \eqref{eq:utildef} to only depend on the agents' decisions locally. Given a utility function $\U$, we can define a new utility structure as $\U^{\rm{mc}}(a_i, a_{-i}) \doteq \U(a) - \U(\emp_i, a_{-i})$, sometimes referred to as the \emph{marginal contribution utility}\cite{marden2013distributed}. Following Eq. \eqref{eq:utildef}, this corresponds to the following expression.
\begin{equation}
\label{eq:mcdef}
    \U^{\rm{mc}}(a_i, a_{-i}) = \sum_{r \in a_i} \f_r(|a|_r),
\end{equation}
where $\f_r(j) \equiv \wf_r(j) - \wf_r(j - 1)$. From the assumptions on $\wf_r$, we have that $\f_r$ is non-negative and non-increasing and that $\f_r(1) = \wf_r(1)$. We observe that now $\U^{\rm{mc}}$ only depends on the resource utilization on the resources selected by agent $i$ while also satisfying the following preference equivalence.
\begin{fact}
For a given joint action $\alpha \in \ac$, the best response according to the marginal contribution utility 
\begin{equation}
\arg \max_{a_i \in \ac_i} \U^{\rm{mc}}(a_i, \alpha_{-i}) = \arg \max_{a_i \in \ac_i} \U(a_i, \alpha_{-i})
\end{equation}
is the same with respect to the original utility.
\end{fact}

An immediate consequence is that the best response process is invariant to using the marginal contribution utility. Thus, we instead use $\U^{\rm{mc}}$ in Eq. \eqref{eq:mcdef} throughout the appendix to define the utility structure. Specifically we will refer to $\f_r$ as \emph{utility rules} and $\fw: \ww \to \R^{\N}_{\geq 0}$ with $\fw(\w_r) = \f_r$ as the utility design rather than using $\wf$ and $\wfw$.

\subsection{Price of Anarchy and $\infty$-round walks}

Note that resource allocation games are isomorphic to \emph{potential games}, and, as such, the limit points from any best response process must necessarily be in the set of Nash equilibrium $\NE \subseteq \ac$ of the game. Thus the limit points $\lim_{\k \to \infty} \sol$ of the $\k$-round walk are Nash equilibrium as well. We consider $\ne \in \NE$ to be a Nash equilibrium if any unilateral deviations are not preferable by any agent, or
\begin{equation}
    \ne_i \in \arg \max_{a_i \in \ac_i} \U(a_i, \ne_{-i}) \ \text{ for all } i \in \p.
\end{equation}
Measuring the quality of Nash equilibrium is done through the classical metric of \emph{price of anarchy} as follows.
\begin{equation}
    \poa(G) = \frac{\min_{a \in \NE}{\W(a)}}{\max_{a \in \mathcal{A}}{\W(a)}}.
\end{equation}
We similarly define $\poa(\ww, \wfw) = \inf_{\G \in \setgm} \poa(\G)$ mirroring Eq. \eqref{eq:effsetG}. The price of anarchy is a well understood metric, with a host of results on its characterization, complexity, and design \cite{roughgarden2009intrinsic}. As solutions of the $\infty$-round walk must also be Nash equilibrium, we have that $\lim_{\k \to \infty} \sol \subseteq \NE$. However, this inclusion may be strict, as not every Nash equilibrium may be reachable from the $\k$-round walk considered in Algorithm \ref{alg:cap2}. But in the next theorem, we show equivalence of price of anarchy and the efficiency of the $\infty$-round walk. Thus in the subsequent sections, we can use previous results in the literature on price of anarchy to quantify the efficiency of the $\infty$-round walk.

\begin{lemma}
\label{lem:poblesspoa}
    Let $\ww$ be a set of welfare rules. The efficiency of the $\infty$-round walk is equivalent to the price of anarchy
    \begin{equation}
    \label{eq:pobeqpoa}
        \pob(\ww, \wfw; \infty) = \poa(\ww, \wfw),
    \end{equation}
    if the utility rules $\f_{\ell} = \fw(\w_\ell)$ for $1 \leq \ell \leq m$ are non-increasing, and $\pob(\ww, \wfw; \infty) \geq \poa(\ww, \wfw)$ otherwise.
\end{lemma}
\begin{proof}
For ease of notation, we remove the subscript of $\f^{\rm{mc}}_{\ell}$ as $\f_{\ell}$. Since $\lim_{\k \to \infty} \sol \subseteq \NE$, by definition, the efficiency guarantee of $\pob(\ww, \wfw; \infty) \geq \poa(\ww, \wfw)$ must be higher than the guarantee for the total set of Nash equilibrium.

Now we show the $\pob(\ww, \wfw; \infty) \leq \poa(\ww, \wfw)$ by a game construction $\G$, in which a Nash equilibrium with the efficiency arbitrarily close to $\poa(\ww, \fw)$ can result from a one-round walk. Let $\varepsilon_1 > 0$ and $\poa^n(\ww, \fw)$ refer to the price of anarchy for the set of games in $\setgm$ that have only $n$ number of agents. Note that $\poa^n(\ww, \fw)$ is non-increasing in $n$ and lower bounded by $0$. Therefore $\poa^n(\ww, \fw)$ is a convergent sequence in $n$ and for any $\varepsilon_1$, there exists an $N_1 \in \N$ such that $\poa^{N_1}(\ww, \fw) - \poa(\ww, \fw) \leq \varepsilon_1$.

Generalizing \cite[Theorem 2]{paccagnan2019utility} to a set of welfare rules provides a characterization of the price of anarchy for $N_1$ agents as $\poa^{N_1}(\ww, \fw) = Q^{-1}$ with
\begin{align}
    &Q = \max_{\paval} \sum_{\substack{1 \leq \ell \leq m, \\ \ay, \xx, \bz}} \w_{\ell}(\bz + \xx) \paval \label{eq:poaLP} \\
    \text{ s.t. } &\sum_{\substack{1 \leq \ell \leq m, \\ \ay, \xx, \bz}} [\ay \f_{\ell}(\ay + \xx) - \bz \f_{\ell}(\ay+\xx+1)] \paval \geq 0 \nonumber \\
    &\sum_{\substack{1 \leq \ell \leq m, \\ \ay, \xx, \bz}} \w_{\ell}(\ay + \xx) \paval = 1 \nonumber \\
    &\paval \geq 0, \nonumber
\end{align}
where $\ay, \xx, \bz \in \N$ with $1 \leq \ay + \xx+ \bz \leq N_1$. We refer to $\paop$ to denote the corresponding optimal variables for $\paval$ of the linear program. We construct a matching game $\G$ as follows. Let $N_2 > N_1$ be the number of agents in the game and $D = N_2 + \ay + \xx - 1$. For each $\ay$, $\xx$, $\bz, \ell$ pair and $1 \leq k \leq D$, we construct a set of resources $\rraxb$ with $\card{\rraxb} = \paop / D$ \footnote{While $\paop / D$ might not be an integer, we can scale $\card{\rraxb}$ uniformly and round to arrive at a game construction with the arbitrarily close efficiency guarantees.}. Each agent $i$ has three actions in its action set $\ac_i = \{\emp_i, \ne_i, \aopt_i\}$. Each agent $i$ selects $\{\rraxb\}_{i \leq k \leq \ay + \xx + i - 1}$ in $\ne_i$ for each pair $\ay, \bz, \xx$, $\ell$. If $\ay + \bz + \xx \leq i \leq N_2$, agent $i$ selects $\{\rraxb\}_{i - \bz \leq k \leq \xx + i - 1}$ in $\aopt_i$ for each pair $\ay, \bz, \xx$. Otherwise for $1 \leq i \leq \ay + \bz + \xx - 1$, $\aopt_i = \emp_i$ and agent $i$ doesn't select any resources in $\aopt_i$. This is shown in Figure \ref{fig:poapob} for one $\ay$, $\xx$, $\bz, \ell$. pair

\begin{figure}[ht]
    \centering
    \includegraphics[width=180pt]{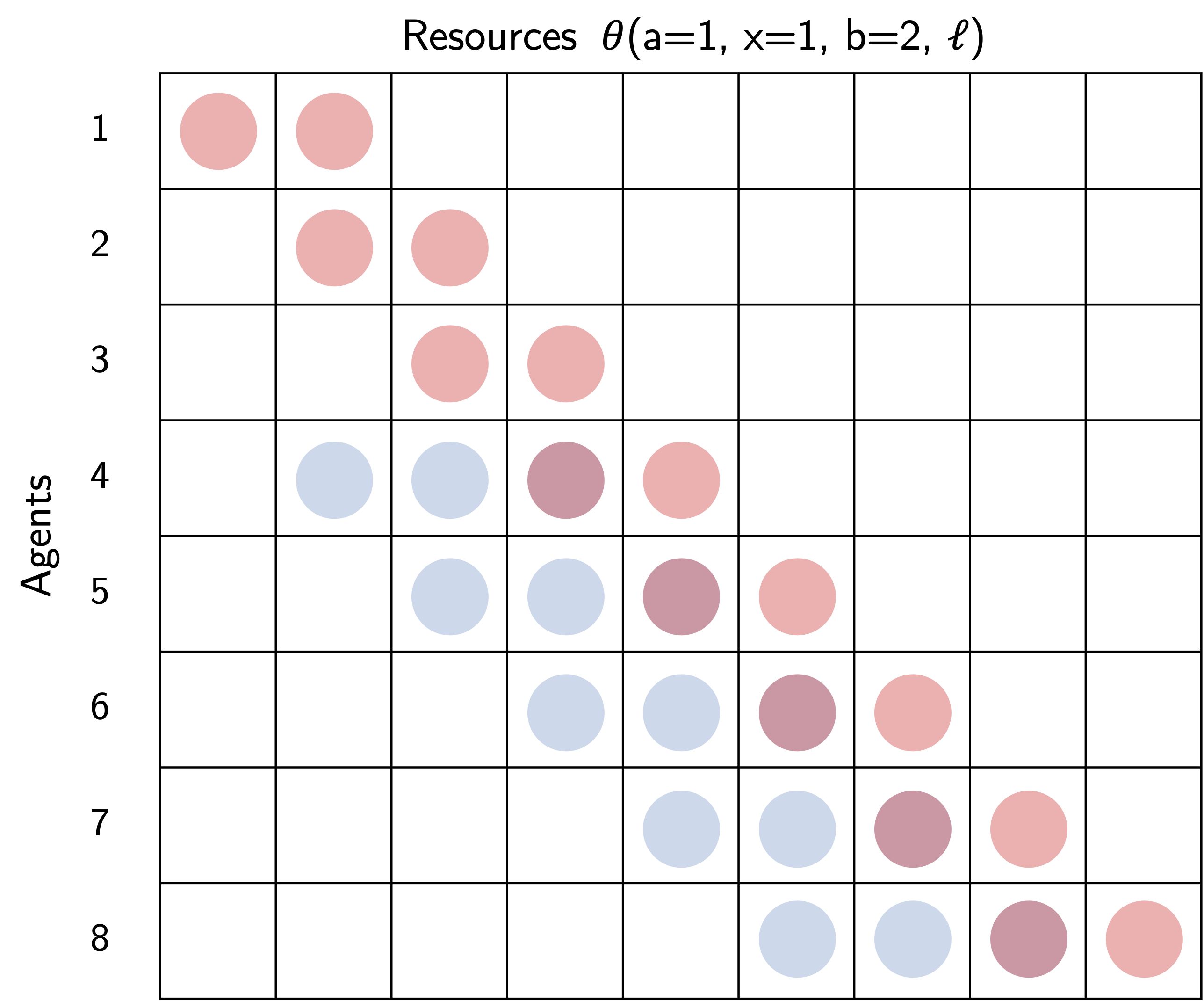}
    \caption{In this figure, rows represent players and columns represent resources. Red circles represent selections in $\ne$ and blue circles represent selections in $\aopt$.}
    \label{fig:poapob}
\end{figure}

We first confirm that the action $\ne$ is indeed a Nash equilibrium. Showing this for the first $\ay + \xx + \bz - 1$ agents is trivial, since no resources are selected in $\aopt_i$. For the rest of the agents, the utility difference of a unilateral deviation to $\aopt_i$ from $\ne_i$ is 
\begin{align*}
    &\U(\ne) - \U(\aopt_i, \ne_{-i}) \\
    &\geq \sum_{r \in \ne_i} \f_r(|\ne|_r) - \sum_{r \in \aopt_i} \f_r(|(\aopt_i, \ne_{-i})|_r) \\
    &\geq \sum_{\ay, \xx, \bz, \ell} [(\ay+\xx) \f_{\ell}(a+x) - \\
    & \quad \ \xx \f_{\ell}(\ay + \xx) - \bz \f_{\ell}(\ay + \xx + 1)] \cdot \card{\rraxb} \\
    &\geq \frac{1}{D} \sum_{\ay, \xx, \bz, \ell} [\ay \f_{\ell}(\ay + \xx) - \bz \f_{\ell}(\ay + \xx + 1)] \paop \\
    &\geq 0.
\end{align*}

The first inequality comes from the definitions of the utility function. The second inequality comes from counting the resources that are selected in the either $\ne_i$ or $\aopt_i$ by the agent in each set of resources in $\rraxb$. The third inequality comes from simplifying. The fourth inequality comes from the fact that since $\paop$ has to satisfy the inequality constraint in Eq. \eqref{eq:poaLP} to be feasible. Similarly, in a one-round walk, the best response for the first $\ay + \xx + \bz - 1$ agents is $\ne_i$. The best response for the other agents during the one-round walk is also $\ne_i$, since
\begin{align*}
    &\U(\ne_{j < i}, \ne_i, \emp_{j > i}) - \U(\ne_{j < i}, \aopt_i, \emp_{j > i}) \\
    &= \sum_{\ay, \xx, \bz, \ell} [\sum_{j=1}^{\ay+\xx} \f_{\ell}(i) - \xx \f_{\ell}(\ay + \xx) - \bz \f_{\ell}(\ay + \xx + 1)] \card{\rraxb} \\
    &\geq \frac{1}{D} \sum_{\ay, \xx, \bz, \ell} [\ay \f_{\ell}(\ay + \xx) - \bz \f_{\ell}(\ay + \xx + 1)] \paop \\
    &\geq 0.
\end{align*}

We use similar arguments as before, where the second inequality comes from the fact that $\f_{\ell}$ is non-increasing. Therefore, the Nash equilibrium $\ne$ is reached from an empty configuration in one-round. Additionally, since $\ne$ is a Nash equilibrium, the resulting action state after any $\k$ rounds can also be itself $\ne$. Therefore in this game, $\pob(\G; \infty) \leq \poa(\G)$. Now we calculate the efficiency of the Nash equilibrium $\W(\ne)$ with respect to $\W(\aopt)$. We have that
\begin{align*}
    \W(\ne) =& \sum_{\ay, \xx, \bz, \ell} \w_{\ell}(\ay + \xx) \cdot \paop \frac{N_2 - 2(\ay + \xx -1)}{N_2} \\
    & + 2 \sum_{\substack{1 \leq \ell \leq m \\ 1 \leq i \leq \ay + \xx -1}} \w_{\ell}(i) \frac{\paop}{N_2} = 1 + \mathrm{O}(\frac{1}{N_2}),
\end{align*}
where, since $\paop$ is feasible, then it satisfies the equality constraint that $\sum_{\ay, \xx, \bz, \ell} \w_{\ell}(\ay + \xx) \paop = 1$. $\mathrm{O}(\frac{1}{N_2})$ reflects that the rest of the terms are on order of $1/N_2$. Similarly, 
\begin{align*}
    \W(\aopt) = \sum_{\ay, \xx, \bz, \ell} \w_{\ell}(\bz + \xx) \cdot \paop \frac{N_2 - 3 (\bz + \xx -1)}{N_2} \\
     + 2 \sum_{\substack{1 \leq \ell \leq m \\ 1 \leq i \leq \bz + \xx -1}} \w_{\ell}(i) \frac{\paop}{N_2} = Q + \mathrm{O}(\frac{1}{N_2}),
\end{align*}
where, since $\paop$ is optimal, then $\sum_{\ay, \xx, \bz, \ell} \w_{\ell}(\bz + \xx) \paop = Q = \poa^{N_1}(\ww, \fw)^{-1}$. For any $\varepsilon_2$, we can choose $N_2$, such that $\mathrm{O}(\frac{1}{N_2}) \leq \varepsilon_2$, so $\poa(\G) \leq \poa^{N_1}(\ww, \fw)^{-1} + \varepsilon_2$. To put everything together, we have that 
\begin{align*}
    &\pob(\ww, \fw; \infty) \leq \pob(\G; \infty) = \poa(\G) \\
    &\leq \poa^{N_1}(\ww, \fw)^{-1} + \varepsilon_2 \leq \poa(\ww, \fw) + \varepsilon_1 + \varepsilon_2,
\end{align*}
and since $\varepsilon_1$ and $\varepsilon_2$ are arbitrary, we have the result.
\end{proof}

\begin{rem*}
A simple corollary from the above proof also shows that $\pob(\ww, \fw; \k) \leq \poa(\ww, \fw)$ for any $k \geq 1$ as well. Thus the efficiency guarantees of $\k$-round walks are upper bounded by the price of anarchy.
\end{rem*}

\subsection{Proof of Proposition \ref{prop:oneroundC}}
The equalities $\pob^*(\ww; 1) = 1 - \cc/2$ in Eq. \eqref{eq:OPTone} and $\pob(\ww, \fwc; 1) = 1/(1 + \cc)$ in Eq. \eqref{eq:CIoneinfty} comes from the results in \cite[Theorem $3$]{konda2024optimal}. The equality $\pob^*(\ww; \infty) = 1 - \cc/e$ in Eq. \eqref{eq:OPTinfty} comes from the fact that $\pob^*(\ww; \infty) = \sup_{\fw} \poa(\ww, \fw)$ by Lemma \ref{lem:poblesspoa} for the set of non-increasing utility rules and that $\sup_{\fw} \poa(\ww, \fw) = 1 -\cc/e$ comes from \cite[Theorem $1$]{chandan2021tractable}. 

Now we show $\pob(\ww, \fwc; \infty) = \left( 1 + \cc \right)^{-1}$. Since any best response with a common interest utility must increase the welfare $\W$, the limiting efficiency $\pob(\ww, \fwc; \infty) \geq \pob(\ww, \fwc; 1) = 1/(1 + \cc)$ is greater than the efficiency of the one-round walk. Since we consider welfare rules $\w_r$ that are submodular, then the utility rules $\fmc_r$ are non-increasing, and we can apply Lemma \ref{lem:poblesspoa} to have that $\pob(\ww, \fwc; \infty) = \poa(\ww, \fwc)$. From applying \cite[Corollary 1]{paccagnan2021utility} with the bent welfare rule $\wbc$ in Eq. \eqref{eq:wbc} gives
\begin{align*}
    &\poa(\ww, \fwc)^{-1} \geq \poa( \wbc, \fwc)^{-1} =  \\
    &1 + \max_{j \geq 1} \left\{ \frac{j}{\wbc(j)} \left[ 2\wbc(j) - \wbc(j-1) - \wbc(j+1) \right]\right\}
\end{align*}
Simplifying the inequality for $b = 1$ gives $\poa(\ww, \fwc) \leq \left( 1 + \cc \right)^{-1}$. Since $\pob(\ww, \fwc; \infty)$ is both upper bounded and lower bounded by $\left( 1 + \cc \right)^{-1}$, we have the result.

\subsection{Proof of Theorem \ref{thm:kroundC}}
\label{subsec:thmkapp}
In this section, we first provide upper bounds on the efficiency metric $\pob^*(\ww; \k)$. To do this, we construct a game $\G$ such that for any utility design $\fw$, rounds $\k \geq 1$, and curvature $\cc$, we have that $\pob(\ww, \fw; \k) \leq  \pob(\G; \k) \leq 1 - \cc/2$. Let $\cc$ be the curvature and consider the bent welfare rule $\wbc$ with $b = 1$ as in Eq. \eqref{eq:wbc} with $\wbc(2) = 2 - \cc$. Additionally, let $\f = \fw(\wbc)$ be the corresponding utility rules for a given utility design. A two-agent game $\G$ is constructed as follows. Let the resource set be $\rr = \rr_1 \cup \rr_2 \cup \rr_3$, where $\rr_j$ is a set of resources such that the ratio of resources satisfies $\card{\rr_1} = \card{\rr_2} = \f(2) \cdot \card{\rr_3}$. If $\f(2)$ is not a whole number, we can scale up $\card{\rr_j}$ uniformly and round $\f(2) \cdot \card{\rr_3}$ to get arbitrarily close to the given ratio. Let $x = \card{\rr_1}$. The action sets for the game construction the agents will be determined by $\f$ according to the following three cases: (\textbf{a}) $0 \leq \f(2) \leq (1-\cc)$, (\textbf{b}) $(1-\cc) \leq \f(2) \leq 1$, and (\textbf{c}) $\f(2) \geq 1$.

For case (a), Agent $1$'s actions are $\ac_1 = \{\emp_1, a_1^1 = \rr_1, a_1^2=\rr_2\}$. Agent $2$'s actions are $\ac_2 = \{\emp_2, a_2^1 = \rr_3, a_2^2 = \rr_1\}$. The optimal allocation is $\aopt = \{a_1^2, a_2^2\}$ resulting in a welfare of $2x$. An allocation that can occur after a one round walk is $\abr = \{a_1^1, a_2^1\}$ resulting in a welfare of $(1 + \f(2))x$. Therefore, $\pob(\G; 1) \leq \frac{(1 + \f(2))x}{2x} \leq 1 - \frac{\cc}{2}$ by assumption of $\f \leq 1-\cc$. Additionally, observe that $\abr$ is a Nash equilibrium and therefore is still the resulting allocation after any number of additional rounds $\k \geq 1$. Therefore $\pob(\ww, \fw; \k) \leq \pob(\G; \k) \leq 1 - \frac{\cc}{2}$ for this case of utility design.

For case (b), Agent $1$'s actions are $\ac_1 = \{\emp_1, a_1^1 = \rr_1, a_1^2=\rr_2\}$. Agent $2$'s actions are $\ac_2 = \{\emp_2, a_2^1 = \rr_3, a_2^2 = \rr_1\}$. The optimal allocation is $\aopt = \{a_1^2, a_2^2\}$ resulting in a welfare of $2x$. An allocation that can occur after a one-round walk is $\abr = \{a_1^1, a_2^2\}$ resulting in a welfare of $\wbc(2) \cdot x$. Therefore, $\pob(\G; 1) \leq \frac{\wbc(2) \cdot x}{2 x} = 1 - \frac{\cc}{2}$. For $\k \geq 2$, there is a best response path that leads to the end state $\abr$. This is achieved by reaching $a' = \{a_1^1, a_2^1\}$ in the first round. As $a'$ is a Nash action, the best response process can remains at $a'$ for $\k-1$ rounds and in the last round, switch to $\abr$. Therefore $\pob(\ww, \fw; \k) \leq \pob(\G; \k) \leq 1 - \frac{\cc}{2}$ for this case.

For case (c), Agent $1$'s actions are $\ac_1 = \{\emp_1, a_1^1 = \rr_1, a_1^2 = \rr_2 \}$. Agent $2$'s actions are $\ac_2 = \{\emp_2, a_2^1 = \rr_1, a_2^2 = \rr_3 \}$. The optimal allocation is $\aopt = \{a_1^2, a_2^2 \}$ resulting in a welfare of $(1 + \f(2))x$. An allocation that can occur after a one round walk is $\abr = \{a_1^1, a_2^1 \}$ resulting in a welfare of $\wbc(2) \cdot x$. Therefore, $\pob(\G; 1) = \frac{\wbc(2) \cdot x}{(1 + \f(2)) x} \leq 1 - \frac{\cc}{2}$ by assumption of $\f(2) > 1$. Additionally, observe that $\abr$ is a Nash equilibrium and therefore is still the resulting allocation after any number of additional rounds. Therefore $\pob(\ww, \fw; \k) \leq \pob(\G; \k) \leq 1 - \frac{\cc}{2}$ for this case.

Since $\f = \fw(\wbc)$ was chosen arbitrarily, we have that the upper bound holds for any utility design and we have shown that $\pob^*(\ww; \k) \leq 1 - \cc/2$. Furthermore, based on our game construction, the efficiency bounds hold even when we relax the class of best response dynamics that we consider. Since the game construction comprises of only two agents, allowing agents to best respond multiple times during a round or best respond out of order of round-robin does not improve the efficiency guarantees that result from the given game $\G$.

\begin{figure}[ht]
    \centering
    \includegraphics[width=100pt]{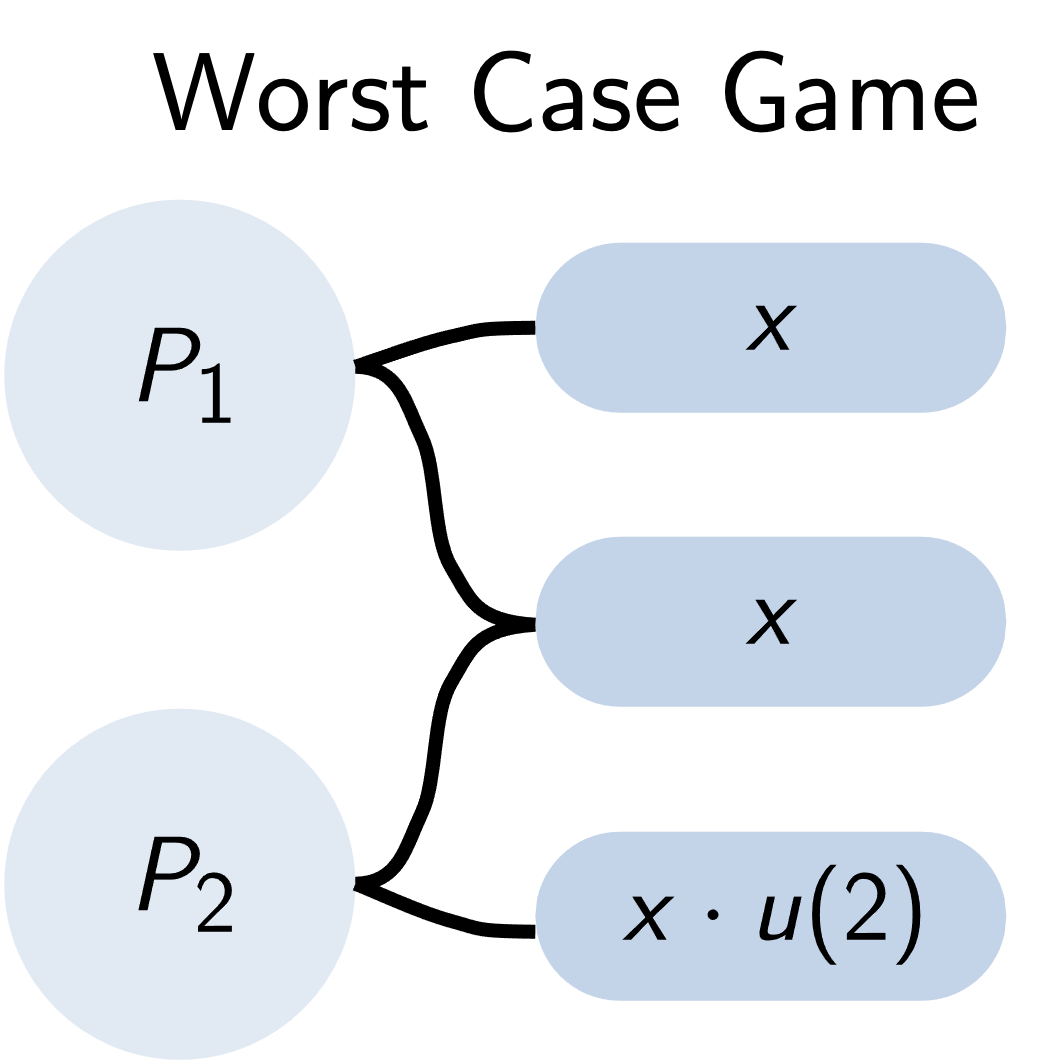}
    \caption{The worst case game construction achieving the $\k$-round walk guarantee dictated in Equation \eqref{eq:klessC}. P1 and P2 represent the two agents and $\card{x}$ and $\card{x \cdot u(2)}$ represent the size of the resources. The black lines represent the selections in the different joint actions by the two agents.}
    \label{fig:worst}
\end{figure}

\subsubsection{Common Interest Utility}
First, we note that since any best response improves the welfare under the common interest utility, we must have that $\pob(\ww, \fwc; \k) \geq \pob(\ww, \fwc; 1) = (1 + \cc)^{-1}$. Now we show that the upper bound $\pob(\ww, \fwc; \k) \leq (1 + \cc)^{-1}$ to complete the equality in Eq. \eqref{eq:klessCCI}. As before, a game $\G$ is constructed such that under the common interest design $\fwc$, $\k \geq 1$, and curvature $\cc$, we have that $\pob(\ww, \fwc; \k) \leq  \pob(\G; \k) \leq (1 + \cc)^{-1}$. Let $\G$ have $n$ players with a resource set $\rr = \rr^{\rm{opt}} \cup \rr^{\rm{both}} \cup \{r^n\}$ with $\card{\rr^{\rm{opt}}} = n$ and $\card{\rr^{\rm{both}}} = n -1$. Each agent $i$ has three actions in its action set $\ac_i = \{\emp_i, \abr_i, \aopt_i\}$. The resources are selected by the agents in the following manner: each resource $r_j^{\rm{opt}} \in \rr^{\rm{opt}}$ is selected by agent $j$ in action $\aopt_j \ni r_j^{\rm{opt}}$ for all $1 \leq j \leq n$; each resource $r_j^{\rm{both}} \in \rr^{\rm{both}}$ is selected by agent $j+1$ in action $\aopt_{j+1} \ni r_j^{\rm{both}}$ and by agent $j$ in action $\abr_{j} \ni r_j^{\rm{both}}$ for all $1 \leq j \leq n-1$; agent $n$ selects the resource $r^n$ in action $\abr_n$. See Figure \ref{fig:worstCI} for a visual representation of this game.

Given a curvature $\cc$, consider two bent welfare rules $\w_1, \w_2 \in \ww$ with curvature $\cc$ such that $\w_1 = \wbc$ with $b=1$ and $\w_2 = \cc \cdot \w_1$. For any $r \in \rr^{\rm{both}} \cup \{r^n\}$, let the corresponding welfare rule be $\w_r = \w_1$ and for any $r \in \rr^{\rm{opt}}$, let the corresponding welfare rule be $\w_r = \w_2$. Under this game construction it can be seen that under $\abr$, each resource $r \in \rr^{\rm{both}} \cup \{r^n\}$ is selected by exactly one agent, resulting in a welfare of $\W(\abr) = n$; also, under $\aopt$, each resource $r \in \rr^{\rm{both}} \cup \rr^{\rm{opt}}$ is selected by exactly one agent, resulting in a welfare of $\W(\aopt) = (n-1)(1+\cc) + \cc$. Assuming that $\abr$ is the joint action that results after $\k$ rounds, we have that $\pob(\G, \k) \leq \frac{n}{(n-1)(1+\cc) + \cc}$. Limiting the number of agents $n \to \infty$ to infinity gives the result. To verify that $\abr$ can result after $\k$ rounds, observe that for agent $1$ selecting $\abr_1$ over $\aopt_1$ results in a higher system welfare. After that, agents $2$ through $n$ are indifferent between $\abr_j$ and $\aopt_j$ given that the previous $i < j$ players have selected $\abr_i$. Therefore, $\abr$ is the resulting allocation after one round. Additionally, $\abr$ is a Nash equilibrium, so after any number of rounds $\k$, the joint action $\abr$ is still the result of a $\k$-round walk.

\begin{figure}[ht]
    \centering
    \includegraphics[width=200pt]{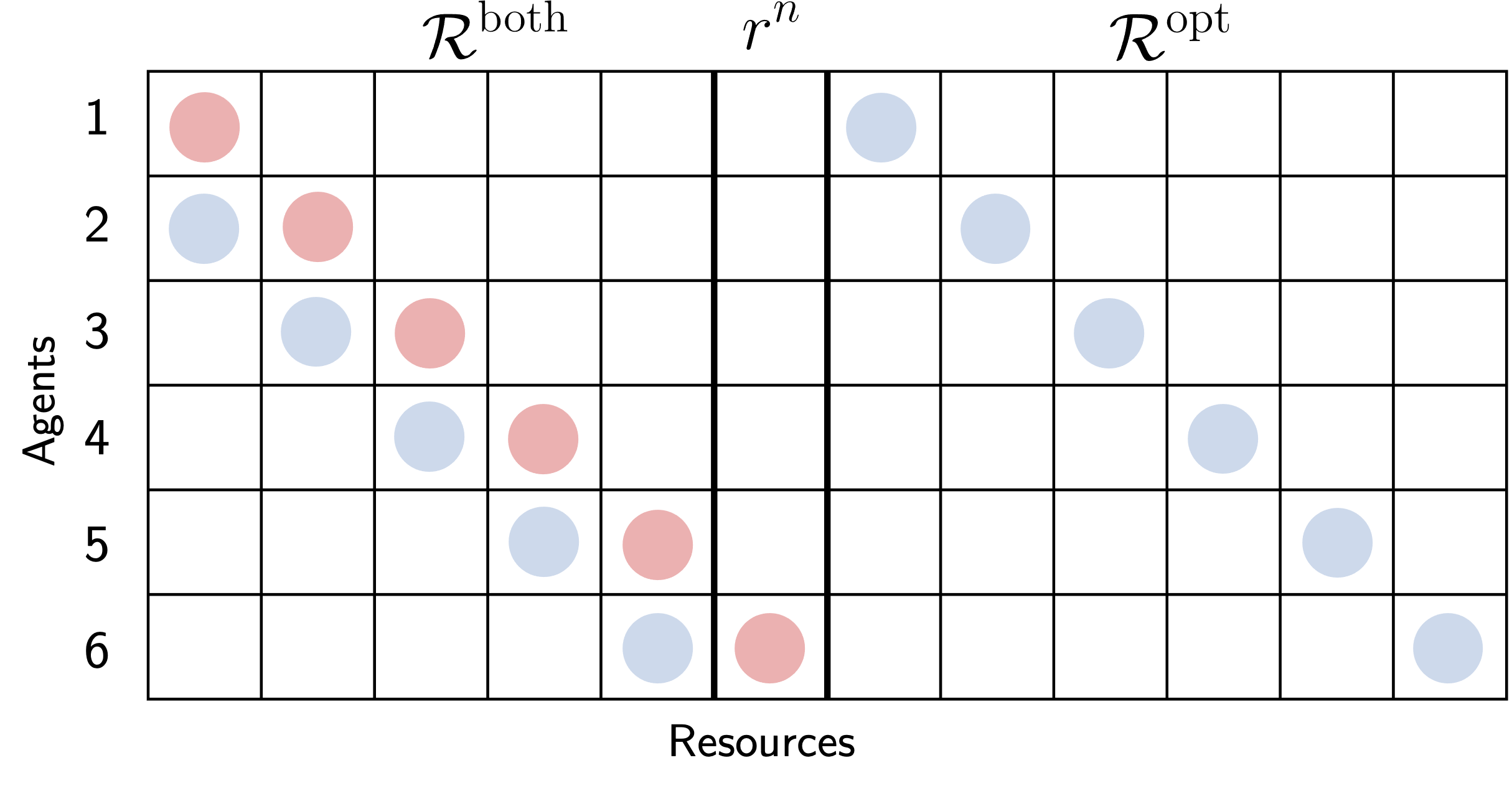}
    \caption{The worst case game construction achieving the $\k$-round walk guarantee dictated in Equation \eqref{eq:klessCCI}. In this figure, rows represent players and columns represent resources. Red circles represent selections in $\abr$ and blue circles represent selections in $\aopt$. }
    \label{fig:worstCI}
\end{figure}

\subsection{Proof of Theorem \ref{thm:submodtrade}}
{ \color{black}
We show the trade-offs in Theorem \ref{thm:submodtrade} that result from considering utility designs that maximize the one-round walk efficiency versus the $\infty$-round walk. We first show the equality in Eq. \eqref{eq:tradfrontone}. From Lemma  \ref{lem:poblesspoa}, we have 
\begin{equation*}
    \pob(\ww, \fw_1^*; \infty) = \poa(\ww, \fw_1^*) \geq \pob(\ww, \fw_1^*; 1) = 1 -\cc/2,
\end{equation*}
since $\fw_1^*$ is a non-increasing utility design, which is shown in \cite[Appendix Section B]{konda2024optimal}. We now show $\poa(\ww, \fw_1^*) \leq 1 -\cc/2$. Consider the bent rule $\wbc$ in Eq. \eqref{eq:wbc} with $b = 1$ and a utility rule $\f_1$ with $\f_1(1) = 1$ and $\f_1(2) = (2-2\cc)/(2-\cc)$. In \cite[Appendix Section C.2]{konda2024optimal}, it was shown that $\f_1$ is the utility rule that maximizes the one-round efficiency for $\ww = \{\wbc\}$. It can be easily verified that $\f_1$ satisfies the assumptions of \cite[Theorem 2]{paccagnan2021utility} and we can derive the price of anarchy as 
\begin{equation*}
    \poa(\wbc, \f_1)^{-1} = {\scriptstyle \max_{1 \leq l \leq j} \left\{\frac{\wbc(l) + j \f_1(j) - l \f_1(j+1)}{\wbc(j)} \right\}.}
\end{equation*}
Under $j=1$ and $l=1$, we have that $\poa(\wbc, \f_1) \leq 1 - \cc/2$. Using the fact that $\poa(\ww, \fw_1^*) \leq \poa(\wbc, \f_1)$, we have the upper bound as well.

We show that $\pob(\ww, \fw_{\infty}; 1) \leq 1 + \frac{(\cc-3)\cc}{(2-\cc)e + \cc}$ in Lemma \ref{lem:bCpoa} and $\pob(\ww, \fw_{\infty}; 1) = 0$ for $\cc = 1$ in Lemma \ref{lem:bCpobbad}. These lemmas are stated below.

\begin{lemma}
\label{lem:bCpoa}
Let the set $\ww$ comprise of all submodular welfare rules $\w$ that have curvature of at most $\cc \in [0, 1]$. Consider $\fw^*_{\infty}$ to be the utility design that achieves the optimal $\pob^*(\ww; \infty)$. Then we have that
\begin{equation}
    \label{eq:tradfrontpoa2}
    \pob(\ww, \fw^*_{\infty}; 1) \leq 1+\frac{\left(\cc-3\right)\cc}{\left(2-\cc\right)e+\cc}.
\end{equation}
\end{lemma}
\begin{proof}
From results in \cite[Lemma 1 iii]{chandan2021tractable}, the utility design that optimizes the price of anarchy $\fw_{\rm{poa}} \doteq \arg \max_{\fw} \poa(\ww, \fw)$ have non-increasing utility rules. Since these utility rules are non-increasing, under Lemma $1$, the resulting asymptotic efficiency $\pob(\ww, \fw_{\rm{poa}}, \infty) = \poa(\ww, \fw_{\rm{poa}})$ and thus $\fw_{\rm{poa}} \equiv \fw^*_{\infty}$ is the utility design that maximizes the asymptotic efficiency.

To construct the upper bound in Eq. \eqref{eq:tradfrontpoa2}, we characterize the one-round efficiency of the asymptotically optimal utility design against the bent welfare rule $\wbc$ for $b = 1$. We have that $\pob(\ww, \fw^*_{\infty}; 1) \leq \pob(\w^{1, \cc}, \fa; 1)$, where $\fa \equiv \fw^*_{\infty}(\w^{1, \cc})$. We can use the linear program given in \cite[Corollary 1]{konda2024optimal} to characterize the one-round efficiency. Under only the bent welfare rule, the constraint for a given $z$, $y$ simplifies to
\begin{align*}
    \beta \cdot \w^{1, \cc}(y) \geq \sum_{i=1}^{y} \fa(i)- z \min_{1 \leq i \leq y+1} \fa(i) + \w^{1, \cc}(z).
\end{align*}
According to \cite[Lemma 1]{chandan2021tractable}, the asymptotically optimal utility rule $\fa$ for the bent welfare rule is given by the following recursive equation
\begin{align*}
    \fa(1) &= 1, \\
    \fa(j+1) &= \max \{j \fa(j) - \rho \w^{1, \cc}(j) + 1, 1 - \cc \},
\end{align*}
with $\rho = (1 - \cc/e)^{-1}$. Now we characterize the binding constraints for $z$. Since $\fa(j) \leq 1$ for all $j$ and $\fa(j) \geq \w^{1, \cc}(j) - \w^{1, \cc}(j-1)$ for $ j \geq 2$, the binding constraint is when $z = 1$. Coupled with the fact that $\fa$ is non-increasing, we can simplify the one-round characterization as
\begin{align*}
    \pob(\w^{1, \cc}, \fa; 1)^{-1} &= \max_{y \geq 1} \left\{\frac{1 + \sum_{i=1}^{y} \fa(i) - \fa(y+1)}{\w^{1, \cc}(y)} \right\} \\
    & \geq \frac{(2 - \cc)e + \cc}{(2 - \cc)(e - \cc)},
\end{align*}
where the second inequality comes from only considering $y = 2$ and solving for $\fa(j)$ explicitly for $j = 2$ and $j = 3$. Taking the reciprocal on both sides and simplifying gives the expression in Eq. \eqref{eq:tradfrontpoa2}.
\end{proof}

\begin{lemma}
\label{lem:bCpobbad}
Let the set $\ww$ comprise of all submodular welfare rules $\w$ that have curvature of at most $\cc=1$. Consider $\fw^*_{\infty}$ to be the utility design that achieves the optimal $\pob^*(\ww; \infty)$. Then $\pob(\ww, \fw^*_{\infty}; 1) = 0$.
\end{lemma}
\begin{proof}
By definition, $\pob(\ww, \fw^*_{\infty}; 1) \geq 0$ must be greater than zero. For the upper bound, we construct a game $\G \in \mathcal{G}_{\ww, \fw^*_{\infty}}$ such that one-round walk efficiency is $\pob(\G; 1) = 0$. Consider a game $\G$ with $n$ players as follows. We partition the resource set as $\rr = \bigcup_{1 \leq j \leq n+1}\rr_j$. Every resource $r \in \rr$ is endowed the local welfare rule $\w_r = \wbc$ as the bent welfare rule with curvature of $\cc = 1$ for some fixed $b \geq 1$, as defined in Eq. The corresponding utility rule is $\fa = \fw^*_{\infty}(\wbc)$ is the following recursive expression from \cite[Lemma 1]{chandan2021tractable}\footnote{The recursive expression found in \cite{chandan2021tractable} is for the utility rule that maximizes the price of anarchy. Since this utility rule is non-increasing, we can apply Lemma \ref{lem:poblesspoa} to translate the results to $\infty$-round walks.},
\begin{align*}
    \fa(1) &= 1 \\
    \fa(j+1) &= \frac{1}{b}[j \fa(j) - \rho^b \min \{j, b\}] + 1,
\end{align*}
with $\rho^b = (1 - \frac{b^b e^{-b}}{b!})^{-1}$. The number of resources in each set is $\card{\rr_1} = v$ and $\card{\rr_{j+1}} \sim v \cdot \fa(j)$ for $1 \leq j \leq n$ and for some $v \geq 0$. If $\fa(j)$ is not a whole number, we can scale $v$ up and round to get arbitrarily close to the correct ratio of resources. Agent $i$ selects $\rr_1 = \abr_i$ and $\rr_{i+1} = \aopt_i$ in each of its actions. It can be verified that $\abr$ is a joint action that can result after a one round walk. Therefore, the efficiency is upper bounded by
\begin{equation*}
\pob(\G; 1) \leq \frac{\W(\abr)}{\W(\aopt)} = \frac{vb}{v \sum_{1 \leq i \leq n}\fa(i)}.
\end{equation*}

Now we show that as we increase $n$, the series $\sum_{1 \leq i \leq n}\fa(i)$ diverges, and the efficiency can get arbitrarily bad as the number of agents increase.
To construct the closed form expression of $\fa(j)$, we construct the following LTV state space system with $\fa(j) := x(t)$
\begin{align*}
    x(t+1) &= A(t) x(t) + s(t) \quad \quad  A(t) = \frac{t}{b} \\
    s(t) &= 1 - \frac{\rho^b}{b} \min(t, b)
\end{align*}
Solving for the solution $x(t)$ using the state transition matrix with the initial condition $x(1) = 1$ results in the following expression
\begin{align*}
    x(t) &= \prod_{\tau = 1}^{t}\frac{\tau}{b} + \sum_{T=1}^{t-1} \big[ \big(1 - \frac{\rho^b}{b} \min(t, b) \big) \prod_{\tau = T+1}^{t-1} \frac{\tau}{b} \big] \\
    &= \frac{t!}{b^t} \bigg(1 + \sum_{T=1}^{t} \frac{b^T}{T!}\big(1 - \frac{\rho^b}{b} \min(t, b) \big) \bigg)
\end{align*}

If $t \geq b$, then
\begin{align*}
    x(t) &= \frac{t!}{b^t} \bigg(1 - (e^b - 1)(\rho - 1) + \sum_{T=t+1}^{\infty} \frac{b^T}{T!}\big(\rho^b - 1 \big) + \\
    &\sum_{T=1}^{b} \frac{b^T}{T!}\frac{\rho^b (b-T)}{b}\bigg) \\
    &=  \frac{t!}{b^t} \sum_{T=t+1}^{\infty} \frac{b^T}{T!}\big(\rho^b - 1 \big) \\
    &\geq (\rho^b - 1) \frac{b}{t + 1} \\
    &\sim \mathrm{O}(\frac{1}{t})
\end{align*}

\noindent The first equality results from splitting the summation and the second equality will be shown later. Since $x(t)$ is on the order of $\frac{1}{t}$, the series $\sum_{i=1}^{N} \fa(i)$ diverges and the claim is shown. Now we verify the equality

\begin{align*}
    \sum_{T=1}^{b} \frac{b^T}{T!}\frac{\rho^b (b-T)}{b} &= (e^b - 1)(\rho^b - 1) - 1 \\ 
    \sum_{T=1}^{b} \frac{b^T(b-T)}{b T!} &= \frac{1}{\rho^b} \big(e^b \rho^b - e^b - \rho^b \big) \\
    \sum_{T=1}^{b}\frac{b^{T}}{T!} - \sum_{T=1}^{b} \frac{b^{T-1}}{(T-1)!} &= \big(e^b - 1 - e^b (1 - \frac{b^b e^{-b}}{b!}) \big) \\
    \frac{b^b}{b!} - 1 &=  \frac{b^b}{b!} - 1
\end{align*}
The last equality results from recognizing the terms on the left hand side as a telescoping sum.
\end{proof}

\subsection{Proof of Theorem \ref{thm:poapobtradeoff}}

To characterize the Pareto optimal frontier in Eq. \eqref{eq:tradeoffsetcov}, we first derive the closed form solution for the one-round walk efficiency $\pob(\wsc, \f; 1)$ specifically for the set covering welfare rule $\wsc$.

\begin{lemma}
\label{lem:setpob}
Let $\ww = \{\wsc\}$, where $\wsc$ is the set covering welfare rule defined in Eq. \eqref{eq:wscdef}, and $\fw(\wsc) = \f$ be the corresponding utility rule. Then the one-round walk efficiency guarantee is
\begin{equation}
    \label{eq:pobwscf}
    \pob(\wsc, \f; 1) =\left[ \sum_{i \in \N}{\f(i)} - \min_{i \in \N}{\f(i)} + 1 \right]^{-1}.
\end{equation}
\end{lemma}
\begin{proof}
Examine the linear program in \cite[Corollary 1]{konda2024optimal} with substituting the set covering welfare defined in Eq. \eqref{eq:wscdef}. Under the substitution, the constraint for a given $z$, $y$ simplifies to 
\begin{equation*}
    \beta \geq \sum_{i=1}^{y} \f(i) - z \min_{1 \leq i \leq y + 1} \f(i) + \min(1, z).
\end{equation*}
We have applied the fact $\wsc(j) = \min(1, j) = 1$ when $j \geq 1$. Observe that the binding constraint occurs when we limit $y \to \infty$ and set $z = 1$ (and not $z=0$ since $\f(1) = 1$, the term $1 - \min_j \f(j) \geq 0$). Under the binding constraint, $\pob(\wsc, \f; 1) = \beta^{-1}$, where $\beta^{-1}$ matches the given expression in Eq. \eqref{eq:pobwscf}.
\end{proof}

\begin{figure}[ht]
    \centering
    \includegraphics[width=250pt]{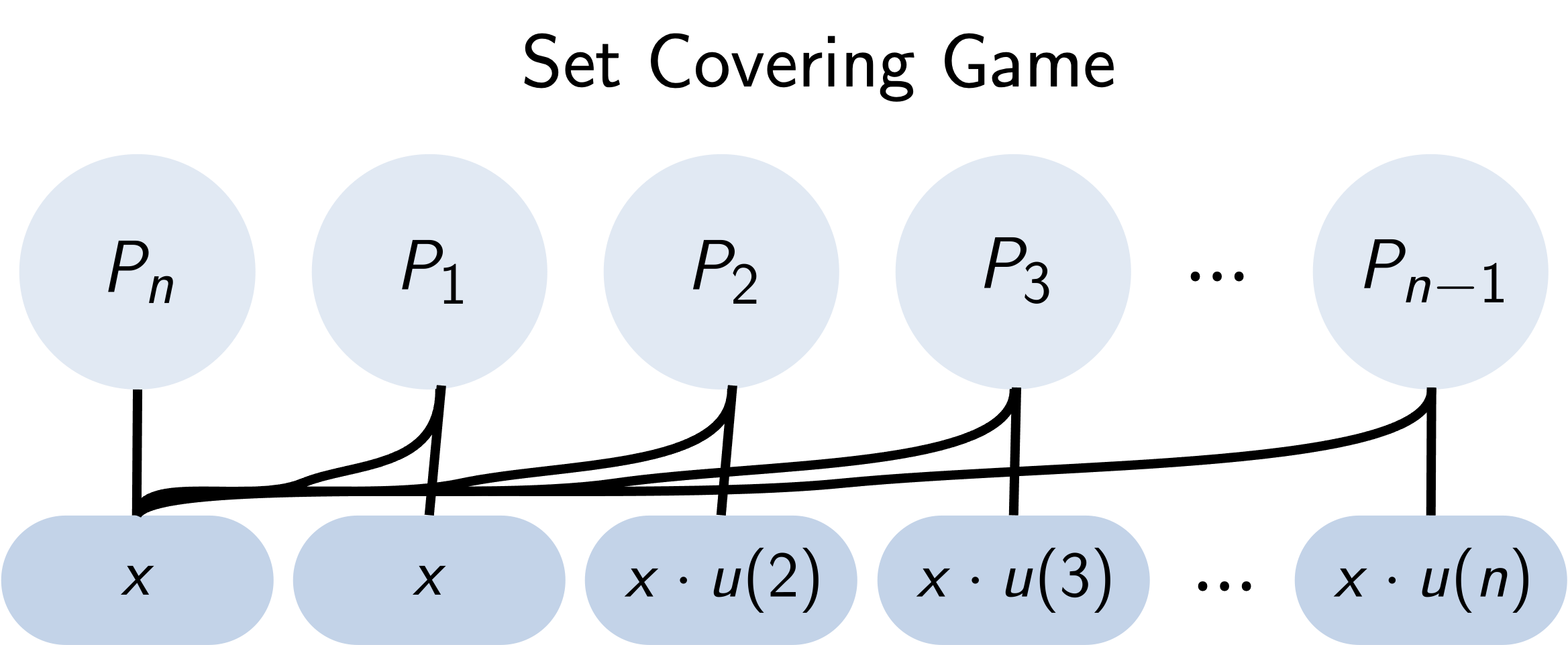}
    \caption{The worst case game construction achieving the one-round walk guarantee dictated by Lemma \ref{lem:setpob}. The agents are represented by circles, each oval represents the number of resources with the set covering welfare, and the black lines represent the agent action selections. In this game, all the agents can either stack on the first resource set or spread out.}
    \label{fig:setcovgame}
\end{figure}

To describe the trade-off, we now provide an explicit expression of Pareto optimal utility rules, i.e., the utility rules $\f$ that satisfy either $\pob(\wsc, \f; 1) \geq \pob(\wsc, \f'; 1)$ or $\pob(\wsc, \f; \infty) \geq \pob(\wsc, \f'; \infty)$ for all $\f'\neq \f$.

\begin{lemma} \label{lem:implicitfboundar}
For a given $\ipoa \geq 0$, a utility rule $
\fxx$ that satisfies $\pob(\wsc, \f; \infty) \geq 1/(1+\ipoa)$ while maximizing $\pob(\wsc, \f; 1)$ is defined as in the following recursive formula:
\begin{equation}
\label{eq:recurf}
\begin{aligned}
    \fxx(1) &= 1  \\
    \fxx(j+1) &= \max\{j \fxx(j) - \ipoa, 0\}.
\end{aligned}
\end{equation}
\end{lemma}
\begin{proof}
We only consider utility rules that are non-increasing in this paper. Under this assumption, from Lemma \ref{lem:poblesspoa}, we have that $\pob(\wsc, \f; \infty) = \poa(\wsc, \f)$. According to Corollary $2$ in \cite{paccagnan2019utility}, the price of anarchy for $n$ agents can be written as 
\begin{equation*}
    \frac{1}{\poa^n(\wsc, \f)} = 1 + \\ \max_{1 \leq j \leq n - 1} \{j\f(j)-\f(j+1), (n-1)\f(n)\}.
\end{equation*}
We define the following constant for each possible utility rule.
\begin{equation}
\label{eq:chiconst}
    \ipoa_{\f} = \max_{1 \leq j \leq n - 1} \{j\f(j)-\f(j+1), (n-1)\f(n)\}.
\end{equation} 
For $\f$ to be Pareto optimal, we claim that $ j \f(j) - \f(j+1) = \ipoa_{\f}$ must hold for all $j$. Consider any other $\f'$ with $\ipoa_{\f} = \ipoa_{\f'}$. It follows that $\poa^n(\wsc, \f) = \poa^n(\wsc, \f') = 1/(1+\ipoa_{
\f})$. By induction, we show that $\f(j) \leq \f'(j)$ for all $j$. The base case is satisfied, as $1 = \f(1) \leq \f'(1) = 1$. Under the assumption $\f(j) \leq \f'(j)$, we also have that 
\begin{equation}
    j\f(j) - \ipoa_{\f} = \f(j+1) \leq \f'(j+1) = j\f'(j) - \ipoa^j_{\f},
\end{equation}
where $\ipoa^j_{\f'} = j\f'(j) - \f'(j+1) \leq \ipoa_{\f'}$ by definition in Eq. \eqref{eq:chiconst}, and so $\f(j) \leq \f'(j)$ for all $j$. Therefore the summation $\sum_{i \in \N}{\f(i)} - \min_{i \in \N}{\f(i)}$ in Eq. \eqref{eq:pobwscf} is diminished when $\f(i) \leq \f'(i)$ for all $i$ and $\pob(\wsc, \f; 1) \geq \pob(\wsc, \f'; 1)$, proving our claim. As $\f$ must satisfy $\f(j) \geq 0$ for all $j$ to be a valid utility rule, $\f(j+1)$ is set to be $\max \{j\f(j) - \ipoa, 0\}$. Then we get the recursive definition for the maximal $\fxx$ in Eq. \eqref{eq:recurf} when we limit $n \to \infty$. Finally, we note that for infinite $n$, $\ipoa < \frac{1}{e-1}$ is not achievable, as shown in \cite{gairing2009covering}.
\end{proof}

\noindent With the two previous lemmas, we can move to proving Theorem \ref{thm:poapobtradeoff}. We first characterize a closed form expression of the maximal utility rule $\fxx$, which is given in Lemma \ref{lem:implicitfboundar}. We fix $\ipoa$ so that $\pob(\wsc, \fxx; \infty) = \frac{1}{\ipoa+1}$ = $Q$. To calculate the expression for $\fxx$ for a given $\ipoa$, a corresponding time varying, discrete time system to Eq. \eqref{eq:recurf} is constructed as follows.
\begin{align*}
    x(t+1) &= t x(t) - \ipoa, \\
    y(t) &= \max \{x(t), 0\}, \\
    x(1) &= 1,
\end{align*}
where $y(t) \equiv \fxx(j)$ corresponds to the utility rule. Solving for the explicit solution for $y(t)$ using the state-transition matrix gives
\begin{align*}
    y(1) &= 1 \\
    y(t) &= \max \Big[ \prod_{\ell=1}^{t-1} \ell - \ipoa \big( \sum_{\tau=1}^{t-2} \prod_{\ell=\tau+1}^{t-1} \ell \big) - \ipoa , 0 \Big] \ \ t > 1.
\end{align*}
Simplifying the expression and substituting for $\fxx(j)$ gives 
\begin{equation*}
    \fxx(j) = \max \Big[ (j-1)!(1 - \ipoa \sum_{\tau=1}^{j-1} \frac{1}{\tau!} \big), 0 \Big] \ \ \ j \geq 1.
\end{equation*}

Substituting the expression for the maximal $\fxx$ into Eq. \eqref{eq:pobwscf} gives the one round efficiency. Notice that for $\ipoa \geq \frac{1}{e-1}$, $\lim_{j \to \infty} \fxx(j) = 0$, and therefore $\min_j{\fxx(j)} = 0$. Shifting the variables $j' = j + 1$, we get the statement in Eq. \eqref{eq:tradeoffsetcov}.

\begin{IEEEbiography}[{\includegraphics[width=1in,height=1.25in,clip,keepaspectratio]{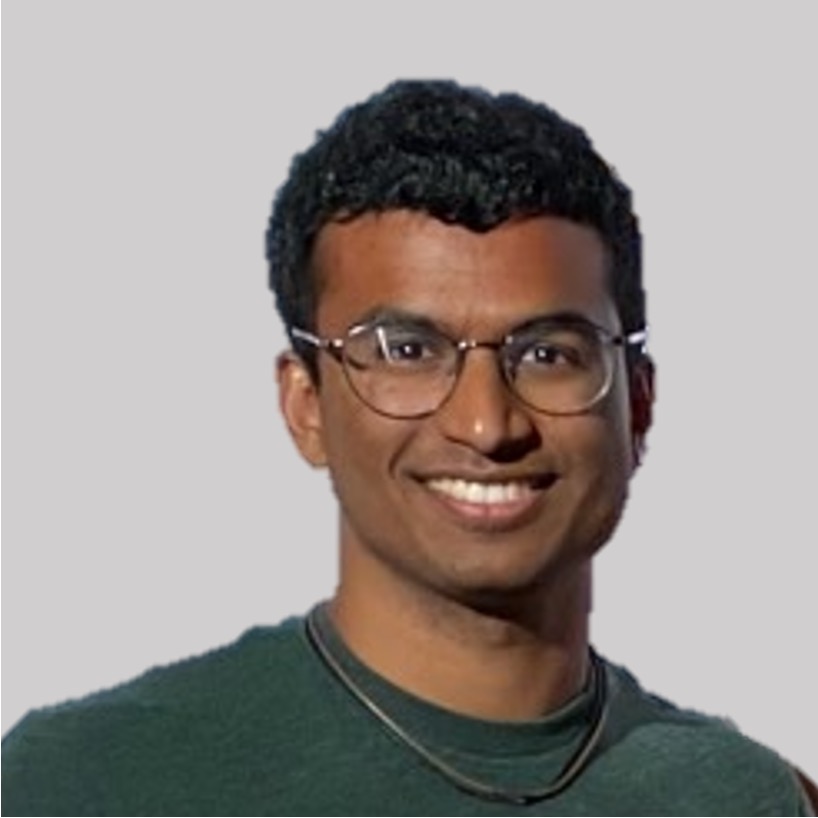}}]{Rohit Konda} received a B.S in Biomedical Engineering in 2018 and an M.S in Electrical and Computer Engineering in 2019 from Georgia Tech. He is currently a Ph.D Candidate in the Department of Electrical and Computer Engineering at the University of California, Santa Barbara. His research interests include distributed optimization and game-theoretic applications in multi-agent systems.
\end{IEEEbiography}

\begin{IEEEbiography}[{\includegraphics[width=1in,height=1.25in,clip,keepaspectratio]{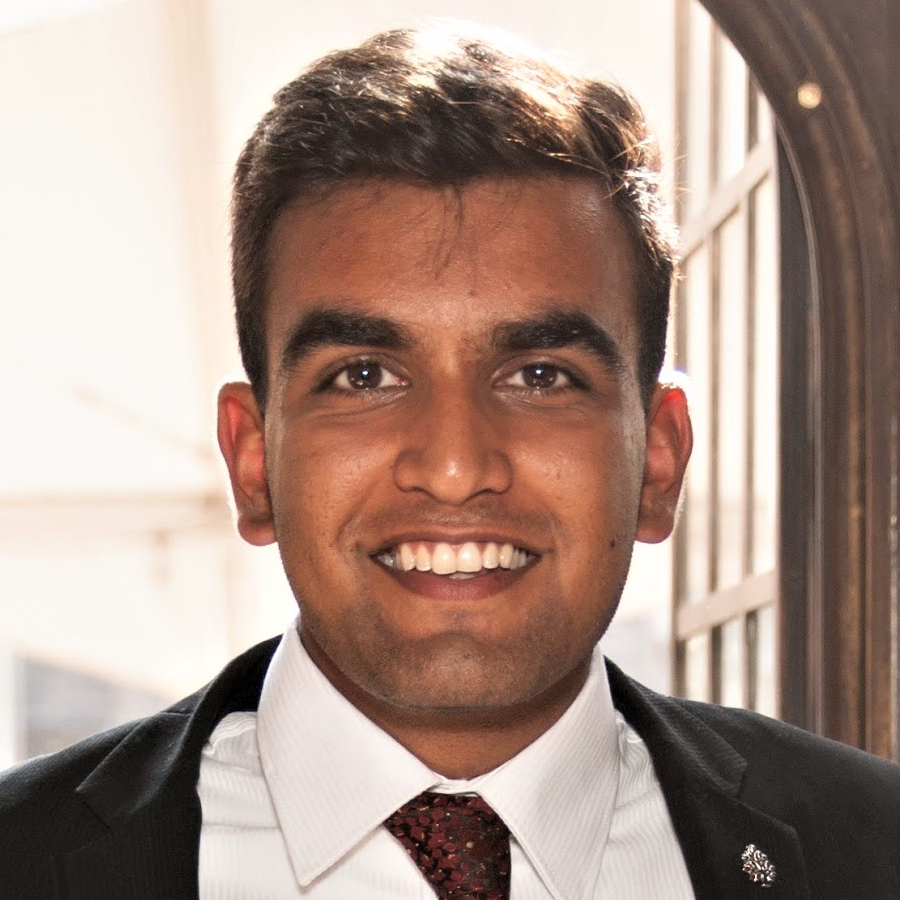}}]{Rahul Chandan} received his B.A.Sc. degree in Engineering Science, and Electrical and Computer Engineering from the University of Toronto, Toronto, ON, Canada, in June 2017. He received his Ph.D. degree in 2022, under the supervision of Jason R. Marden with the Electrical and Computer Engineering Department, University of California, Santa Barbara, Santa Barbara, CA, USA. His research interests include the application of game theoretic and control methods to the analysis and design of multi-agent systems.
\end{IEEEbiography}

\begin{IEEEbiography}[{\includegraphics[width=1in,height=1.25in,clip,keepaspectratio]{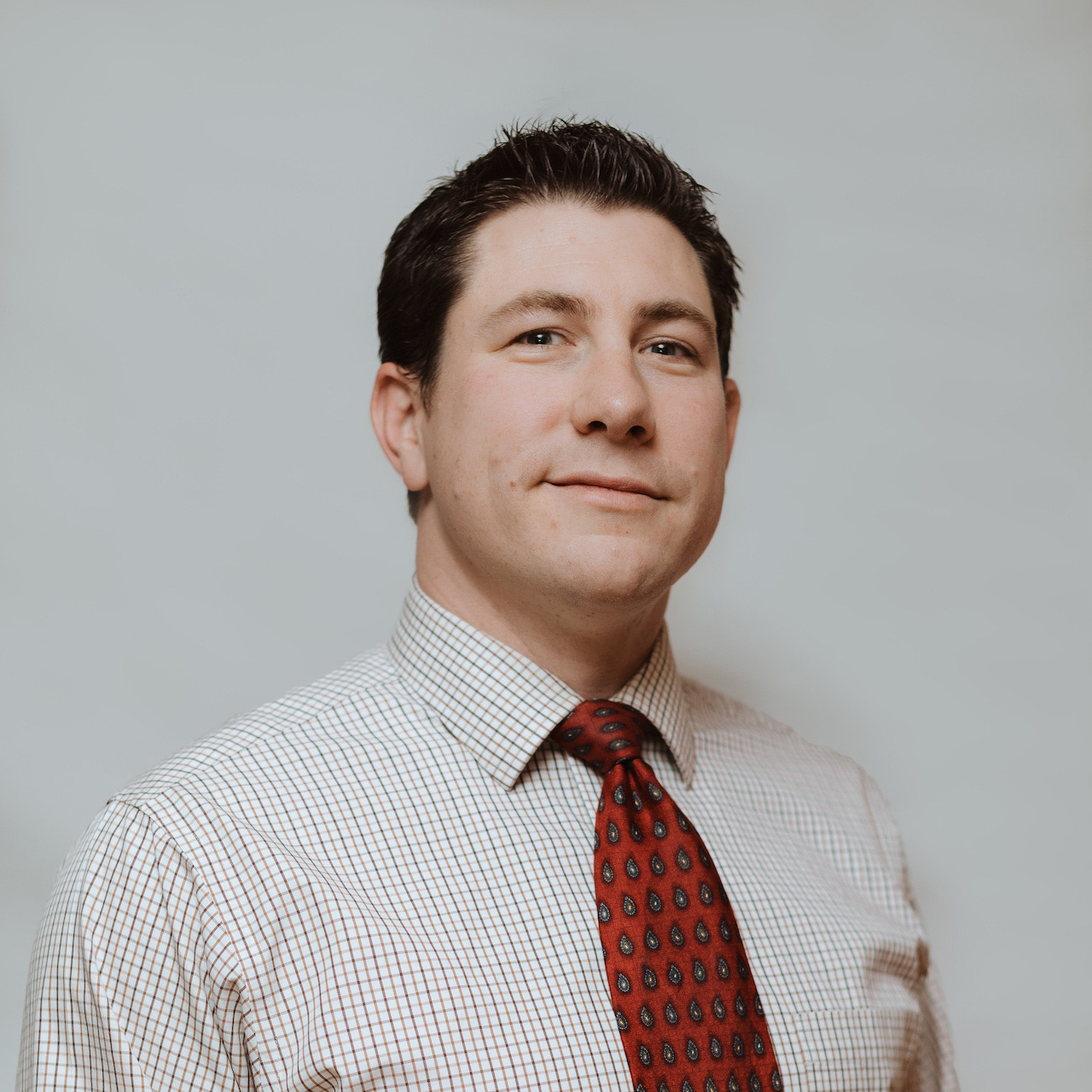}}]{David Grimsman} received the B.S. degree in electrical and computer engineering from Brigham Young University, Provo, UT, USA, in 2006 as a Heritage Scholar, with a focus on signals and systems. He received the M.S. degree in computer science from Brigham Young University in 2016. He received the Ph.D. degree in 2021, under the supervision of Jason R. Marden and Joao P. Hespanha, with the Electrical and Computer Engineering Department, University of California, Santa Barbara, Santa Barbara, CA, USA. He is currently an Assistant Professor with the Department of Computer Science in Brigham Young University.

He worked with BrainStorm, Inc. for several years as a Trainer and
IT Manager. His research interests include multiagent systems, game
theory, distributed optimization, network science, linear systems theory, and security of cyberphysical systems.
\end{IEEEbiography}

\begin{IEEEbiography}[{\includegraphics[width=1in,height=1.25in,clip,keepaspectratio]{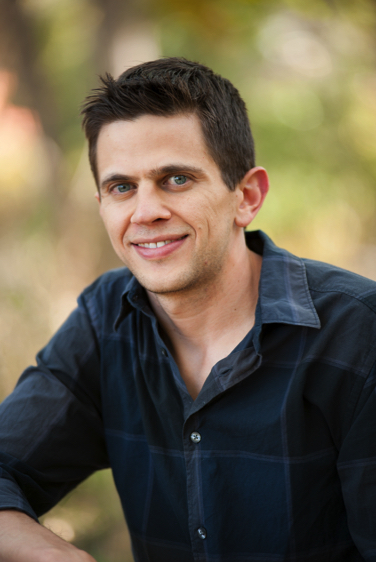}}]{Jason R. Marden} received the B.S. degree in 2001 and the Ph.D. degree in 2007 under the supervision of Jeff S. Shamma from the University of California, Los Angeles (UCLA), Los Angeles, CA, USA, both in mechanical engineering. After graduating from UCLA, he served as a Junior Fellow with the Social and Information Sciences Laboratory, California Institute of Technology until 2010 when he joined the University of Colorado. He is currently a Full Professor with the Department of Electrical and Computer Engineering, University of California, Santa Barbara, CA, USA. His research interests focus on game theoretic methods for the control of distributed multi-agent systems.

Dr. Marden is a recipient of the NSF Career Award in 2014, he was awarded the Outstanding Graduating Ph.D. degree, the ONR Young Investigator Award in 2015, the AFOSR Young Investigator Award in 2012, the American Automatic Control Council Donald P. Eckman Award in 2012, and the SIAG/CST Best SICON Paper Prize in 2015.
\end{IEEEbiography}
}

\end{document}